\documentclass[oneside,reqno]{amsart}
\usepackage{amsmath, amsthm,amssymb,amsfonts,latexsym}
\usepackage[dvips]{graphicx}
\usepackage{color}
\usepackage{array}
\usepackage[hmargin=3cm,vmargin=3.5cm]{geometry}
\usepackage{booktabs, multicol, multirow}
\usepackage{verbatim}
\newtheorem{prop}{Proposition}
\newtheorem{cor}{Corollary}
\newtheorem{lemma}{Lemma}
\newtheorem{thm}{Theorem}

\newtheorem{assump}{Standing Assumption}
\usepackage{tikz}
\usepackage{caption}
\usepackage{subcaption}
\usetikzlibrary{trees,shapes,snakes}
\newcommand{\cblue}{\color{black}}

\DeclareMathOperator{\sgn}{sgn}

\begin{document}

\title{A multi-asset investment and consumption problem with transaction costs}
\author{David Hobson \and Alex Tse \and Yeqi Zhu}
\thanks{David Hobson: Department of Statistics, University of Warwick, Coventry, CV4 7AL, UK. D.Hobson@warwick.ac.uk; \\
Alex Tse: Cambridge Endowment for Research in Finance, Judge Business School, University of Cambridge, Cambridge, CB2 1AG, UK. S.Tse@jbs.cam.ac.uk \\
Yeqi Zhu: Credit Suisse, London, UK. (The opinions expressed in the paper are those of the author and not of Credit Suisse.) Yeqi.Zhu@credit-suisse.com}
\date{\today}

\begin{abstract}
In this article we study a multi-asset version of the Merton investment and consumption problem with proportional transaction costs. In general it is difficult to make analytical progress towards a solution in such problems, but we specialise to a case where transaction costs are zero except for sales and purchases of a single asset which we call the illiquid asset.

Assuming agents have CRRA utilities and asset prices follow exponential Brownian motions we
show that the underlying HJB equation can be transformed into a boundary value problem for a first order differential equation. The optimal strategy is to trade the illiquid asset only when the fraction of the total portfolio value invested in this asset falls outside a fixed interval. Important properties of the multi-asset problem (including when the problem is well-posed, ill-posed, or well-posed only for large transaction costs) can be inferred from the behaviours of a quadratic function of a single variable and another algebraic function.


\end{abstract}

\maketitle


\section{Introduction}
\label{sect:intro}

In one of his seminal works, Merton~\cite{Merton:69} considers a portfolio and consumption problem faced by a price-taking agent in a continuous-time stochastic financial model consisting of a risk-free bond and a risky asset. The agent is assumed to have the objective of maximising the expected discounted utility from consumption
over an infinite horizon. In a model in which the single risky asset follows an exponential Brownian motion with constant parameters and the agent has constant relative risk aversion, Merton shows that the optimal behaviour is to consume at a rate which is proportional to wealth, and to invest a constant fraction of wealth in the risky asset. The result generalises easily to multiple risky assets.

Constantinides and Magill~\cite{ConstantinidesMagill:76} were the first to add proportional transaction costs to the model. In a model with a single risky asset they conjectured the form of the optimal strategy, namely it is optimal to keep the fraction of wealth invested in the risky asset in an interval. Subsequently Davis and Norman~\cite{DavisNorman:90} gave a precise statement of the result and showed how the solution could be expressed in terms of local times. The optimal behaviour is to trade in a minimal fashion so as to keep the variables (cash wealth, wealth in the risky asset) in a wedge-shaped region in the plane, and this is achieved by sales and purchases of the risky asset in the form of singular stochastic controls.

The approach in Davis and Norman~\cite{DavisNorman:90} is to write down the Hamilton-Jacobi-Bellman (HJB) equation, and to characterise the candidate value function as a solution to this equation. Shreve and Soner~\cite{ShreveSoner:94} reproved the results of \cite{DavisNorman:90} using viscosity solutions and gave several extensions. These approaches remain the main methods for solving portfolio optimisation problems with transaction costs, although recently a different technique based on shadow prices has been proposed, see Guasoni and Muhle-Karbe~\cite{GuasoniMuhleKarbe:12} for a users' guide.
Kallsen and Muhle-Karbe~\cite{KallsenMuhleKarbe:10}, Choi et al~\cite{ChoiSirbuZitkovic:13} and Herczegh and Prokaj~\cite{HerczeghProkaj:15} use the dual approach to characterise the solution to the problem with transaction costs and one risky asset.

The results in Davis and Norman~\cite{DavisNorman:90} are limited to a single risky asset, and it is of great interest to understand how they generalise to multiple risky assets.
In his survey article on consumption/investment problems with transaction costs Cadenillas~\cite[page 65]{Cadenillas:00} says that `most results in this survey are limited to the case of only one bond and only one stock. It is then important to see if these results can be extended to cover a realistic number of stocks'. Although there has been some progress since that paper was published, similar sentiments are echoed in recent papers by Chen and Dai~\cite[page 2]{ChenDai:13}: `most of the existing theoretical characterisations of the optimal strategy are for the single risky-asset case. In contrast there is a relatively limited literature on the multiple risky-asset case' and Guasoni and Muhle-Karbe~\cite[page 194]{GuasoniMuhleKarbe:12}: `In sharp contrast to frictionless models, passing from one to several risky assets is far from trivial with transaction costs \ldots multiple assets introduce novel effects, which defy the one-dimensional intuition'. In summary therefore, there is great interest in both theoretical and numerical results on the multi-asset case, and this paper can be considered as a contribution to that literature.

In the multi-asset case, and on the computational side, Muthuraman and Kumar~\cite{MuthurmanKumar:06} use a process of policy improvement to construct a numerical solution for the value function and the associated no-transaction region, Collings and Haussmann~\cite{CollingsHausmann:99} derive a numerical solution via a Markov chain approximation for which they prove convergence, {\cblue and Dai and Zhong~\cite{DaiZhong:10} use a penalty method to obtain numerical solutions}. On the theoretical front Akian {\em et al}~\cite{AkianMenaldiSulem:95} show that the value function is the unique viscosity solution of the HJB equation (and provide some numerical results in the two-asset case) and Chen and Dai~\cite{ChenDai:13}
identify the shape of the no-transaction region in the two-asset case. Explicit solutions of the general problem remain very rare.

One situation when an explicit solution is possible is the rather special case of uncorrelated risky assets, and an agent with constant absolute risk aversion. In that case the problem decouples into a family of optimisation problems, one for each risky asset, see Liu~\cite{Liu:04}.
Another setting for which some progress has been made is the problem with small transaction costs, see Whalley and Wilmott~\cite{WhalleyWilmott:97}, {\cblue Janecek and Shreve~\cite{JanecekShreve:04}, Bichuch and Shreve~\cite{BichuchShreve:13}, Soner and Touzi~\cite{SonerTouzi:13}, and, for a recent analysis in the multi-asset case, Possama\"{i} et al~\cite{PossamaiSonerTouzi:15}. These papers use an expansion method to provide asymptotic formulae for the optimal strategy, value function and no-transaction region.}

{\cblue Our focus is on optimal investment/consumption problems, but there is a parallel literature on optimal investment problems involving maximising expected utility at a distant terminal horizon, see, for example, Dumas and Luciano~\cite{DumasLuciano:91} for an explicit solution in the one-asset case and Bichuch and Guasoni~\cite{BichuchGuasoni:16} for recent work in a setting similar to ours with liquid and illiquid assets}.

In this paper we consider the problem with a risk-free bond and two risky assets. Transactions in the first risky asset are costless, but
transactions in the second risky asset, which we term the illiquid asset, incur proportional costs. {\cblue This is also the setting of a recent paper by Choi~\cite{Choi:16}.}
More generally, we may have several risky assets on which no transaction costs are payable. By a mutual fund theorem, this general case can be reduced to the
case with a single liquid, risky asset.

This paper is an extension of Hobson et al~\cite{HobsonTseZhu:16} which considers a similar problem with a bond and an illiquid asset but with no other risky assets\footnote{\cblue This paper can also be viewed as a development of the results of Hobson and Zhu~\cite{HobsonZhu:14}. The model in Hobson and Zhu includes both a liquid risky asset and an illiquid asset, but assumes that transaction cost on sales of the illiquid asset is infinite. This case might be called the ``perfectly illiquid'' case: the illiquid asset can be sold, but not bought, and the problem is an optimal liquidation problem. This paper extends Hobson and Zhu~\cite{HobsonZhu:14} to allow for finite transaction costs and purchases of the illiquid asset.}. Many of the techniques of \cite{HobsonTseZhu:16} carry over to the wider setting of this paper. {\cblue (Similarly, the paper of Choi~\cite{Choi:16} extends the work of Choi et al~\cite{ChoiSirbuZitkovic:13} to include a risky liquid asset.)} However, since there are fewer parameters when the financial market includes just one risky asset, the problem in \cite{HobsonTseZhu:16} is significantly simpler and much more amenable to a comparative statics analysis. In contrast, this paper treats the multi-asset problem which has proved so difficult to analyse in full generality, albeit in a rather special case. The multi-asset setting brings new challenges and complicates the analysis.

It is straightforward to write down the Hamilton-Jacobi-Bellman (HJB) equation for our problem. The value function is a function of four variables (wealth in liquid assets, price of the illiquid asset, quantity of illiquid asset held, time) and satisfies a HJB equation which is second order, non-linear and subject to value matching and smooth fit at a pair of unknown free boundaries. (The smooth fit turns out to be of second order.)
Our first achievement is to show that the problem of finding the free boundaries and the value function can be reduced to the study of a boundary crossing problem for a family of solutions to a class of first order ordinary differential equations parametrised by the initial values. This allows us to characterise precisely the parameter combinations for which the problem is well-posed (Theorem~\ref{thm:wellposed}), and in those cases to give an expression for the value function (Theorem~\ref{thm:valfun}). These results extend Choi et al~\cite{ChoiSirbuZitkovic:13} and Hobson et al~\cite{HobsonTseZhu:16} to the case of multiple risky assets.

{\cblue As mentioned above, Choi~\cite{Choi:16} studies a similar problem. The main difference between this paper and Choi~\cite{Choi:16} is that we analyse the HJB equation, whereas Choi takes the dual approach and studies shadow prices. Choi~\cite{Choi:16}[Remark 2.2, Assumption 2.4] assumes that the corresponding two-asset problem with zero transaction costs is well-posed, and hence the problem with liquid and illiquid assets is well-posed whatever the value of transaction costs. In contrast, in addition to the unconditionally well-posed case, we also consider the case where the problem is ill-posed for zero and small transaction costs, but well-posed for large transaction costs. (Note that analysis of this situation is beyond the scope of approaches which rely on expansions in a (small) transaction cost parameter.) In fact we show (Corollary~\ref{cor:wellposed}) that the problem is well-posed for sufficiently large transaction costs provided the problem is well-posed when the liquid asset is omitted. We call the case when the problem is well-posed only for large transaction costs the conditionally well-posed case.}

Our second achievement is to make definitive statements about the comparative statics for the problem. We focus on the boundaries of the no-transaction wedge and the certainty equivalent value of the holdings in the illiquid asset. Amongst other results, we prove (see Theorem~\ref{thm:compstat} and Corollary~\ref{cor:compstat} for precise statements) that as the drift on the illiquid asset improves, the agent aims to keep a larger fraction of his total wealth in the illiquid asset, in the sense that the critical ratios at which sales and purchases take place are increasing in the drift.
Conversely, as the agent becomes more impatient, the agent keeps a smaller fraction of wealth in the illiquid asset. Further, we prove (Theorem~\ref{thm:compstat2} and Corollary~\ref{cor:compstat2}) that as the drift on the illiquid asset improves, or as the agent becomes less impatient, the certainty equivalent value of the holdings in the illiquid asset increases. See Section~\ref{sec:compstat} for a more detailed discussion.

The remainder of the paper is as follows. In the next section we formulate the problem. In Section~\ref{sect:hjb_derive} we derive the HJB equation and give heuristics showing how it can be converted to a free boundary value problem involving a first order differential equation. Then we can state our main results on the existence of a solution (Section~\ref{sec:main}). In Section~\ref{sect:fbp} we discuss the various cases which arise. In Section~\ref{sec:compstat} we discuss the comparative statics of the problem, before Section~\ref{sec:conc} concludes. Materials on the solution of the free boundary value problem, the verification argument for the HJB equation, and other lemmas on the analysis of solutions of the differential equations are relegated to the appendices.

\section{The problem}
\label{sect:problem}

The economy consists of one money market instrument paying constant interest rate $r>0$ and two risky assets, one of which is liquidly traded while the other one is illiquid. There are no transaction costs associated with trading in the liquid asset. Meanwhile, trading in the illiquid asset incurs a proportional transaction cost $\lambda\in[0,\infty)$ on purchases and $\gamma\in[0,1)$ on sales, where not both $\lambda$ and $\gamma$ are zero. Let $(S,Y)=(S_{t},Y_{t})_{t\geqslant 0}$ be the price processes of the liquid and illiquid assets respectively. The price dynamics are given by
\[
(S_{t},Y_t) = \left( S_{0}\exp\left((\mu-\frac{\sigma^{2}}{2})t+\sigma B_{t}\right),
Y_{0}\exp\left((\alpha-\frac{\eta^{2}}{2})t+\eta W_{t}\right) \right)
\]
where $(B,W)$ is a pair of Brownian motions with correlation coefficient $\rho\in(-1,1)$. Write $\beta:=(\mu-r)/\sigma$ and $\nu:=(\alpha-r)/\eta$ for the Sharpe ratio of the liquid and illiquid asset respectively.

Let $\Theta_{t}$ be the number of units of the illiquid asset held by an agent at time $t$. Then $\Theta_{t}=\Theta_{0}+\Phi_{t}-\Psi_{t}$ where $\Phi=(\Phi_{t})_{t\geqslant 0}$ and $\Psi=(\Psi_{t})_{t\geqslant 0}$ are both increasing, non-negative processes representing the cumulative units of purchases and sales respectively of the illiquid asset. Let $C=(C_{t})_{t\geqslant 0}$ be the non-negative consumption rate process of the agent and $\Pi=(\Pi_{t})_{t\geqslant 0}$ be the cash value of holdings in the risky liquid asset. We assume $\Theta$, $C$ and $\Pi$ are progressively measurable and right-continuous. If $X=(X_{t})_{t\geqslant 0}$ is the total value of the liquid instruments (cash and the liquid risky asset) then, assuming transaction costs are paid in cash, {\cblue and consumption is from the cash account},
\begin{align*}
dX_{t}&= r(X_t - \Pi_{t})dt +  \frac{\Pi_{t}}{S_{t}}dS_{t} -C_{t}dt-Y_{t}(1+\lambda)d\Phi_{t}+Y_{t}(1-\gamma)d\Psi_{t} \\
&=\left[(\mu-r)\Pi_{t}+rX_{t}-C_{t}\right]dt-Y_{t}(1+\lambda)d\Phi_{t}+Y_{t}(1-\gamma)d\Psi_{t}+\sigma \Pi_t dB_{t}.
\end{align*}

We say that a portfolio $(X,\Theta)$ is solvent at time $t$ if its instantaneous liquidation value is non-negative, that is
\begin{align*}
X_{t}+\Theta_{t}^{+}Y_{t}(1-\gamma)-\Theta_{t}^{-}Y_{t}(1+\lambda)\geqslant 0.
\end{align*}
A consumption/investment strategy $(C,\Pi,\Theta)$ is said to be admissible if the resulting portfolio is solvent at the current time and at all the future time points. Write $\mathcal{A}(t,x,y,\theta)$ for the set of admissible strategies with initial time-$t$ value $(X_{t-}=x,Y_{t}=y,\Theta_{t-}=\theta)$.

We assume the agent has a CRRA utility function with risk aversion parameter $R \in (0,\infty) \setminus \{1\}$. His objective is to find an optimal strategy which maximises the expected lifetime discounted utility from consumption. The problem is thus to find
\begin{align}
V(x,y,\theta)=\sup_{(C,\Pi,\Theta)\in\mathcal{A}(0,x,y,\theta)}\mathbb{E}\left(\int_{0}^{\infty}e^{-\delta s}\frac{C_{s}^{1-R}}{1-R}ds\right)
\label{eq:thevalfun}
\end{align}
where $\delta$ is the agent's subjective discount rate.

We will call $X_t+ \Theta_t Y_t$ the paper wealth of the agent. In our parametrisation a key quantity will be $P_{t}:=\frac{\Theta_{t}Y_{t}}{X_{t}+\Theta_{t}Y_{t}}$, the proportion of paper wealth invested in the illiquid asset.

\section{The HJB equation and a free boundary value problem}
\label{sect:hjb_derive}

\subsection{Deriving the HJB equation}
Let
\begin{align*}
\mathcal{V}(x,y,\theta,t)=\sup_{(C,\Pi,\Theta)\in\mathcal{A}(t,x,y,\theta)}\mathbb{E}\left(\int_{t}^{\infty}e^{-\delta s}\frac{C_{s}^{1-R}}{1-R}ds\right)
\end{align*}
be the forward-starting value function from time $t$. Inspired by the analysis in the classical case involving a single risky asset only, we postulate that the value function has the form
\begin{align}
\mathcal{V}(x,y,\theta,t)&=e^{-\delta t}V(x,y,\theta) = \Upsilon 
\frac{e^{-\delta t}(x+y\theta)^{1-R}}{1-R}G\left(\frac{y\theta}{x+y\theta}\right)
\label{eq:valfun}
\end{align}
for some strictly positive function $G$ to be determined and $\Upsilon$ a convenient scaling constant which will help simplify the HJB equation. We take $\Upsilon = \left(\frac{b_{1}}{R b_{4}}\right)^{-R}$ where $b_1$ and $b_4$ are constants to be defined in Section~\ref{ssec:reduce} below in terms of the financial parameters associated with the underlying problem. For the present we assume that $G$ is smooth and use heuristic arguments to derive a characterisation of the candidate value function. Later we will outline a verification argument that this candidate value function coincides with the solution of the corresponding optimal investment/consumption problem, and therefore deduce the necessary smoothness properties of $\mathcal{V}$ and $G$.

Building on the intuition developed by Constantinides and Magill~\cite{ConstantinidesMagill:76} and Davis and Norman~\cite{DavisNorman:90} we expect that the optimal strategy of the agent is to trade the illiquid asset only when $P_{t}$ falls outside a certain interval $[p_{*},p^{*}]$ to be identified. Due to the solvency restriction, we must have $-\frac{1}{\lambda}\leqslant P_{t}\leqslant \frac{1}{\gamma}$ and $[p_{*},p^{*}] \subseteq [-\frac{1}{\lambda}, \frac{1}{\gamma}]$. Whenever $P_{t}<p_{*}$, the agent purchases the illiquid asset to bring $P_{t}$ back to $p_{*}$. Hence for an initial position $(x,\theta)$ such that $p=\frac{y\theta}{x+y\theta}<p_{*}$, the number of units of illiquid asset to be purchased is given by $\phi=\frac{xp_{*}-(1-p_{*})y\theta}{y(1+\lambda p_{*})}$ such that $\frac{y(\theta+\phi)}{x+y(\theta+\phi)-y(1+\lambda)\phi}=p_{*}$. The value function does not change on this transaction, and hence we deduce that for $-\frac{1}{\lambda}\leqslant p<p_{*}$,
\begin{align*}
(x+y\theta)^{1-R}G(p)=[x+y(\theta+\phi)-y(1+\lambda)\phi]^{1-R}G(p_{*})
\end{align*}
and in turn
\begin{align}
G(p)&=\left(\frac{1+\lambda p}{1+\lambda p_{*}}\right)^{1-R}G(p_{*})=A_{*}(1+\lambda p)^{1-R}
\label{eq:buyregime}
\end{align}
where $A_{*}:=G(p_{*})(1+\lambda p_{*})^{R-1}$. Similar consideration leads to the conclusion that
\begin{align}
G(p)=\left(\frac{1-\gamma p}{1-\gamma p^{*}}\right)^{1-R}G(p^{*})=A^{*}(1-\gamma p)^{1-R}
\label{eq:sellregime}
\end{align}
for $p^{*}<p\leqslant\frac{1}{\gamma}$ where $A^{*}:=(1-\gamma p^{*})^{R-1}G(p^{*})$.

Consider $M=(M_{t})_{t\geqslant 0}$ defined via
\begin{align*}
M_{t}:=\int_{0}^{t}e^{-\delta s}\frac{C_{s}^{1-R}}{1-R}ds+e^{-\delta t}V(X_{t},Y_{t},\Theta_{t}).
\end{align*}
We expect $M$ to be a supermartingale in general, and a martingale under the optimal strategy. Suppose $V$ is $C^{2 \times 2 \times 1}$. Then applying Ito's lemma we find
\begin{align*}
e^{\delta t} dM_{t}&=\frac{C_{t}^{1-R}}{1-R}dt+V_{x}dX_{t}+\frac{1}{2}V_{xx}d[X]_{t} +V_{y}dY_{t}+\frac{1}{2}V_{yy}d[Y]_{t}+V_{\theta}d\Theta_{t}+V_{xy}d[X,Y]_{t}-\delta V dt \\
&=\left(\frac{C_{t}^{1-R}}{1-R}-V_{x}C_{t}+\frac{\sigma^{2}}{2}V_{xx}\Pi_{t}^{2}+((\mu-r) V_{x}+\sigma\eta\rho V_{xy}Y_{t})\Pi_{t}+r V_{x}X_{t}+\alpha V_{y} Y_{t}+\frac{\eta^{2}}{2}V_{yy}Y_{t}^{2}-\delta V\right)dt \\
&\qquad+(V_{\theta}-(1+\lambda)V_{x}Y_{t})d\Phi_{t}+(V_{x}Y_{t}(1-\gamma)-V_{\theta})d\Psi_{t}+\sigma V_{x}\Pi_{t}dB_{t}+\eta V_{y} Y_{t} dW_{t}.
\end{align*}
Further assume $V$ is strcitly increasing and concave in $x$. Then on maximising the drift term with respect to $C_{t}$ and $\Pi_{t}$ and setting the resulting maxima to zero, we obtain the HJB equation over the no-transaction region:
\begin{align}
\frac{R}{1-R}V_{x}^{1-1/R}+rx V_{x}+\alpha yV_{y}+\frac{\eta^{2}}{2}y^{2}V_{yy}-\frac{(\beta V_{x}+\eta\rho yV_{xy})^{2}}{2V_{xx}}-\delta V=0.
\label{eq:hjb}
\end{align}

\subsection{Reduction to a first order free boundary value problem}
\label{ssec:reduce}

Define the auxiliary parameters $b_{1}$, $b_{2}$, $b_{3}$ and $b_{4}$ as
\begin{align*}
b_{1}=\frac{2\left[\delta-r(1-R)-\frac{\beta^{2}(1-R)}{2R}\right]}{\eta^{2}(1-\rho^{2})},\quad b_{2}=\frac{\beta^{2}-2R\eta\rho\beta+\eta^{2}R^{2}}{\eta^{2}R^{2}(1-\rho^{2})},\quad b_{3}=\frac{2(\nu-\beta\rho)}{\eta(1-\rho^{2})}, \quad b_{4}=\frac{2}{\eta^{2}(1-\rho^{2})}.
\end{align*}
It will turn out that the optimal investment and consumption problem depends on the original parameters only through these
auxiliary parameters and the risk aversion level $R$.

Here $b_1$ plays the role of a `normalised discount factor', which adjusts the discount factor to allow for numeraire growth effects and for investment opportunities in the transaction-cost free risky asset.
$b_{4}$ is a simple function of the `idiosyncratic volatility' of the illiquid asset.
The parameter $b_3$ is the `effective Sharpe ratio, per unit of idiosyncratic volatility' of the illiquid asset.
The parameter $b_2$ is the hardest to interpret: essentially it is a
nonlinearity factor which arises from the multi-dimensional structure of the problem.
Note that $b_{2}=1+\frac{1}{1-\rho^{2}}\left(\frac{\beta}{\eta R}-\rho\right)^{2}\geqslant 1$.

In the sequel we will work with the following assumption.
\begin{assump}
\label{ass:sa}
Throughout the paper we assume $b_{1}>0$, $b_2>1$ and $b_{3}>0$.
\end{assump}
The rationale for imposing $b_{1}>0$ is that $b_1>0$ is necessary to ensure well-posedness of the Merton problem in the absence of the illiquid asset. (If $R<1$ and $b_{1}\leqslant 0$, the value function is infinite for the Merton problem. Conversely, if $R>1$ and $b_{1}\leqslant 0$, then for every admissible strategy the expected discounted utility of consumption equals $-\infty$. {\cblue If the Merton problem is ill-posed in the absence of the illiquid asset, then our problem is necessarily ill-posed.})

In contrast, the assumption $b_3 > 0$ is not necessary. However, the advantage of working with a positive effective Sharpe ratio of the illiquid asset ($b_{3}>0$) is that the no-transaction wedge is contained in the first two quadrants of the $(x,y\theta)$ plane. The assumption $b_3>0$ reduces the number of cases to be considered in our analysis, and facilitates the clarity of the exposition, but the methods and results developed in this paper can be extended easily to the case of an illiquid asset with negative effective Sharpe ratio\footnote{\cblue If $b_3=0$ the agent chooses never to invest in the illiquid asset. In this case agent closes any initial position in $Y$ at time zero and thereafter the problem reduces to a standard Merton problem with the single risky asset $S$ and no transaction costs.}.

The case $b_2 = 1$ is rather special and we exclude it from our analysis.
One scenario in which we naturally find $b_2 = 1$ is if $\beta = 0 = \rho$. In this case there is neither a hedging motive, nor an investment motive
for holding the liquid risky asset\footnote{\cblue More generally, the position in the liquid asset $S$ is a combination of an investment position to take advantage of the expected excess returns in $S$ and a hedging position to offset the risk of the position in the illiquid asset $Y$. If $\frac{\beta}{\eta r} = \rho$ then when $X=0$ these terms exactly cancel. In particular, if the half-line $X=0$ is inside the no-transaction region, then since consumption takes place from the cash account, if ever $X=0$ then wealth can only go negative. Then the subspace $X\leq0$ is absorbing, and no further purchases of the liquid asset are ever made.}. Essentially then, the investor can ignore the presence of the liquid risky asset, reducing the dimensionality of the problem. This problem is the subject of \cite{HobsonTseZhu:16}. If $b_2=1$ then the solution $n$ we define in the next paragraph may pass through singular points. See
Choi et al~\cite{ChoiSirbuZitkovic:13} or Hobson et al~\cite{HobsonTseZhu:16} for a discussion of some of the issues.

We adopt the same transformation as \cite{HobsonTseZhu:16} to reduce the order of the HJB equation. Recall the relationship between $\mathcal{V}$ and $G$ in \eqref{eq:valfun} and the definition $p = \frac{y\theta }{x + y\theta}$. Away from $p=1$, set $h(p)=\sgn(1-p)|1-p|^{R-1}G(p)$, $w(h)=p(1-p)\frac{dh}{dp}$, $W(h)=\frac{w(h)}{(1-R)h}$, let $N=W^{-1}$ be the inverse function to $W$ and set $n(q)=|N(q)|^{-1/R}|1-q|^{1-1/R}$. Then, we show in Appendix~\ref{subsect:hjb} that \eqref{eq:hjb} can be transformed into a first order differential equation
\begin{align}
n'(q)=O(q,n(q))
\label{eq:theode}
\end{align}
where
\begin{align}
O(q,n)=\frac{(1-R)n}{R(1-q)}-\frac{2(1-R)^{2}qn/R}{2(1-R)(1-q)\left[(1-R)q+R\right]-\varphi(q,n)-\sgn(1-R)\sqrt{\varphi(q,n)^{2}+E(q)^{2}}}
\label{eq:formO}
\end{align}
with
\begin{align*}
\varphi(q,n)&:=b_{1}(n-1)+(1-R)(b_{3}-2R)q+(2-b_{2})R(1-R), \\
E(q)^{2}&:=4R^{2}(1-R)^{2}(b_{2}-1)(1-q)^{2}.
\end{align*}

Define the quadratic
\begin{equation}
\label{eq:mdef}
 m(q)  :=  \frac{R(1-R)}{b_{1}}q^{2}-\frac{b_{3}(1-R)}{b_{1}}q+1
\end{equation}
and the algebraic function
\begin{equation}
\label{eq:elldef}
\ell(q):=m(q)+\frac{1-R}{b_{1}}q(1-q)+\frac{(b_{2}-1)R(1-R)}{b_{1}}\frac{q}{(1-R)q+R} .
\end{equation}
Note that $m$ has a turning point (a minimum if $R<1$ and a maximum if $R>1$) at $\frac{b_3}{2R} := q_M$ and set $m_M := m(q_M) = 1 - \frac{b_3^2(1-R)}{4 b_1 R}$.

The following are the key properties of the function $O$. They are special cases of a more complete set of properties given in Lemma~\ref{lemma:list} below.
\begin{lemma}
\begin{enumerate}
\item $O(q,n)$ can be extended to $q=1$ by continuity on $(1-R)n<(1-R)\ell(1)$;
\item $O(q,n)=0$ if and only if $n=m(q)$;
\item For given $R$ and $q$ the sign of $O(q,n)$ depends only on the signs of $n-m(q)$ and $\ell(q)-n$.
\end{enumerate}
\label{lem:shortlist}
\end{lemma}

Now we apply the same transformations which took \eqref{eq:hjb} to \eqref{eq:theode} to the value function on the purchase and sale regime. For $-\frac{1}{\lambda}\leqslant p<p_{*}$, $G(p)=A_{*}\left(1+\lambda p\right)^{1-R}$ as given by \eqref{eq:buyregime}.
Then
\[ w(h) = p(1-p) \frac{dh}{dp} = p(1-p)(1-R)h \left[ \frac{\lambda}{1+\lambda p} + \frac{1}{1-p} \right] = (1-R)h \left[ \frac{p(1+\lambda)}{1 + \lambda p} \right] \]
and $|1-W(h)| = \frac{|1-p|}{1+\lambda p} = \left( \frac{A_*}{|h|} \right)^{1/(1-R)}$. It follows that
$n(q)=(A_{*})^{-1/R}$. This expression holds for $-\frac{1}{\lambda}\leqslant p<p_{*}$ on which $q=W(h)=\frac{(1+\lambda)p}{1+\lambda p}$. The equivalent range in $q$ is thus given by $q<q_{*}:=\frac{(1+\lambda)p_{*}}{1+\lambda p_{*}}$.  Similarly on the sale region we have $n(q)=(A^{*})^{-1/R}$ for $q>q^{*}:=\frac{(1-\gamma)p^{*}}{1-\gamma p^{*}}$.

The $C^{2\times2\times1\times1}$ smoothness of the original value function $\mathcal{V}$ now translates into $C^{1}$ smoothness of the transformed value function $n$. Hence we are looking for a continuously differentiable function $n$ and boundary points $(q_{*},q^{*})$ solving \eqref{eq:theode} on $q\in(q_{*},q^{*})$ with  $n(q)=(A_{*})^{-1/R}$ for $q \leq q_*$ and $n(q)=(A^{*})^{-1/R}$ for $q \geq q^*$. First order smoothness of $n$ at the boundary points forces $n'(q_{*})=n'(q^{*})=0$. By Lemma \ref{lem:shortlist}, $n'(q)=O(q,n(q))=0$ if and only if $n(q)=m(q)$. Hence the free boundary points must be given by the $q$-coordinates where $n$ intersects the quadratic $m$. The free boundary value problem now becomes solving $n'(q)=O(q,n(q))$ on $q\in(q_{*},q^{*})$ subject to $n(q_{*})=m(q_{*})$ and $n(q^{*})=m(q^{*})$.

As an example, suppose $R<1$ and $m_M>0$. Fix $u \in (0,q_M)$. Then the solution to \eqref{eq:theode} started at $(u,m(u))$ is decreasing; we are interested in when this solution crosses $m$ again; call this point $\zeta(u)$. Then we have a family of solutions $(n_u(q))_{u \leq q \leq \zeta(u)}$ to \eqref{eq:theode} with $n(u) = m(u)$ and $n(\zeta(u))=m(\zeta(u))$. The solution we want is the one which is consistent with the given transaction costs. Our approach is based on the same idea as in \cite{HobsonTseZhu:16}. Let $\xi=\frac{\lambda+\gamma}{1-\gamma}>0$ be the round-trip transaction cost. Suppose for now $1\notin[p_{*},p^{*}]$ and in turn $1\notin[q_{*},q^{*}]$. Exploiting the relationships that $q_{*}=\frac{(1+\lambda)p_{*}}{1+\lambda p_{*}}$ and  $q^{*}=\frac{(1-\gamma)p^{*}}{1-\gamma p^{*}}$, we have
\[ \ln(1+\xi) = \ln(1+\lambda) - \ln(1-\gamma) = \int_{p_{*}}^{p^{*}}\frac{dp}{p(1-p)}-\int_{q_{*}}^{q^{*}}\frac{dq}{q(1-q)} . \]
Then, using the definitions of $w$, $N$ and $O$,
\begin{eqnarray}
\ln(1+\xi)
&=&\int_{h_{*}}^{h^{*}}\frac{dh}{w(h)}-\int_{q_{*}}^{q^{*}}\frac{dq}{q(1-q)} \nonumber \\
&=&\int_{q_{*}}^{q^{*}}\frac{N'(q)dq}{(1-R)qN(q)}-\int_{q_{*}}^{q^{*}}\frac{dq}{q(1-q)} \nonumber \\
&=&\int_{q_{*}}^{q^{*}}\frac{R}{q(1-R)}\left(\frac{N'(q)}{RN(q)}-\frac{1-R}{R(1-q)}\right)dq \nonumber\\
&=&\int_{q_{*}}^{q^{*}}\left(-\frac{R}{q(1-R)}\frac{O(q,n(q))}{n(q)}\right)dq \label{eq:tci}
\end{eqnarray}
where to get the last line we use the fact that $\frac{O(q,n(q))}{n(q)}=\frac{n'(q)}{n(q)}=\frac{1-R}{R(1-q)}-\frac{1}{R}\frac{N'(q)}{N(q)}$. Hence the required solution from the free boundary value problem is the one such that
\begin{align}
\ln(1+\xi)=\int_{q_{*}}^{q^{*}}\left(-\frac{R}{q(1-R)}\frac{O(q,n(q))}{n(q)}\right)dq
\label{eq:trancost}
\end{align}
holds.

In the case where $1 \in  [p_{*}, p^{*}]$ or equivalently $1 \in [q_{*}, q^{*}]$, the integrals $\int_{p_{*}}^{p^{*}}\frac{dp}{p(1-p)}$ and $\int_{q_{*}}^{q^{*}}\frac{dq}{q(1-q)}$ are not well defined. But it can be shown that \eqref{eq:trancost} still holds using a limiting argument, see Appendix~\ref{app:tcc}.

To summarise, we would like to solve the following:
\begin{quotation}
\noindent \textit{(The free boundary value problem) find a positive function $n(\cdot)$ and a pair of boundary points $(q_{*},q^{*})$ solving
\begin{eqnarray}
n'(q)=O(q,n(q)), && q\in[q_{*},q^{*}] \nonumber \\
n(q_{*})=m(q_{*}), && n(q^{*})=m(q^{*})
\label{eq:freebound}
\end{eqnarray}
and \eqref{eq:trancost}}.
\end{quotation}

In Section \ref{sect:fbp}, we distinguish several different cases and discuss how to construct the solution $(n(\cdot),q_{*},q^{*})$ in each of these cases.

The central role played by the quadratic $m$ is clear from \eqref{eq:freebound}. The function $\ell$ acts as a bound on the feasible solutions to $n'=O(n,q)$, at least for $0<q \leq 1$. Suppose, for example, that $R<1$. Then for $q \in [q_*,q^*]$ we have $n(q) \geq m(q)$ by construction, but also $n(q) \leq \ell(q)$ for $q_* \leq q \leq q^* \wedge 1$. Moreover, the value $\ell(1)$ is crucial in determining when the problem is ill-posed.

\section{Main results}
\label{sec:main}

In Section \ref{sect:hjb_derive} we converted the original HJB equation into the free boundary value problem \eqref{eq:freebound}. Now we argue that, given a solution $(n(\cdot),q_{*},q^{*})$ to \eqref{eq:freebound} we can reverse the transformations and construct a candidate value function. 

Suppose there exists a solution $(n(\cdot),q_{*},q^{*})$ to \eqref{eq:freebound} with $n$ being strictly positive. Define $p_{*}=\frac{q_{*}}{1+\lambda(1-q_{*})}$ and  $p^{*}=\frac{q^{*}}{1-\gamma(1-q^{*})}$. Let $N(q)=\sgn(1-q)n(q)^{-R}|1-q|^{R-1}$, $W=N^{-1}$ and $w(h)=(1-R)hW(h)$. We would like to construct the candidate value function from $G(p) = \sgn(1-p) |1-p|^{1-R} h(p)$ where $h$ solves $\frac{dh}{dp} = \frac{w(h)}{p(1-p)}$.
The main subtlety is that $\frac{w(h)}{p(1-p)}$ is not well-defined at $p=1$. Nonetheless, the definition of $G$ at $p=1$ can be understood in a limiting sense. To this end, we distinguish two different cases based on whether $(q_*-1)$ and $(q^*-1)$ have the same sign or not, or equivalently whether the no-transaction wedge, plotted in $(x, y \theta)$ space, includes the vertical axis $x=0$ (corresponding to $p=1$).
\begin{prop}
(i) Suppose $1\notin[p_{*},p^{*}]$. Define $h(p)$ via
\begin{align}
\int_{N(q_{*})}^{h(p)}\frac{du}{w(u)}=\int_{p_{*}}^{p}\frac{du}{u(1-u)}
\label{eq:def_h_1}
\end{align}
on $p_{*}\leqslant p\leqslant p^{*}$. Then \eqref{eq:def_h_1} is equivalent to
\begin{align}
\int_{h(p)}^{N(q^{*})}\frac{du}{w(u)}=\int_{p}^{p^{*}}\frac{du}{u(1-u)}
\label{eq:def_h_2}
\end{align}
and \eqref{eq:def_h_2} is an alternative definition of $h(p)$.

Let
\begin{align*}
G^{C}(p)=
\begin{cases}
n(q_{*})^{-R}\left(1+\lambda p\right)^{1-R},& p\in[-\frac{1}{\lambda},p_{*}); \\
\sgn(1-p)|1-p|^{1-R}h(p),&p\in[p_{*}, p^{*}]; \\
n(q^{*})^{-R}\left(1-\gamma p\right)^{1-R},& p\in(p^{*},\frac{1}{\gamma}].
\end{cases}
\end{align*}
Then $G^{C}$ is a $C^{2}$ function on $(-\frac{1}{\lambda},\frac{1}{\gamma})$.
Moreover $\frac{(x+y \theta)^{1-R}}{1-R} G^C(\frac{y \theta}{x + y \theta})$ is strictly increasing and strictly concave in $x$.

(ii)
Suppose $1\in[p_{*},p^{*}]$. Define $h(p)$ via
\begin{align*}
\begin{cases}
\int_{N(q_{*})}^{h(p)}\frac{du}{w(u)}=\int_{p_{*}}^{p}\frac{du}{u(1-u)},& p_{*}<p<1;\\
\int_{h(p)}^{N(q^{*})}\frac{du}{w(u)}=\int_{p}^{p^{*}}\frac{du}{u(1-u)},&1<p<p^{*}.
\end{cases}
\end{align*}
Let
\begin{align*}
G^{C}(p)=
\begin{cases}
n(q_{*})^{-R}\left(1+\lambda p\right)^{1-R},& p\in[-\frac{1}{\lambda},p_{*}); \\
\sgn(1-p)|1-p|^{1-R}h(p),& p\in[p_{*}, p^{*}]\setminus\{1\}; \\
n(1)^{-R}e^{-(1-R)a},& p=1;\\
n(q^{*})^{-R}\left(1-\gamma p\right)^{1-R},& p\in(p^{*},\frac{1}{\gamma}]
\end{cases}
\end{align*}
with $a:=-\int_{q_{*}}^{1}\left(\frac{R}{q(1-R)}\frac{O(q,n(q))}{n(q)}\right)dq-\ln(1+\lambda)$. Then $|a|\leqslant \ln(1+\xi)$, and $G^{C}$ is a $C^{2}$ function on $(-\frac{1}{\lambda},\frac{1}{\gamma})$. Moreover $\frac{(x+y \theta)^{1-R}}{1-R} G^C(\frac{y \theta}{x + y \theta})$ is strictly increasing and strictly concave in $x$.
\label{prop:def_g}
\end{prop}
{\cblue Proposition~\ref{prop:def_g} is proved in Appendix~\ref{app:prop}.}

The first pair of main results of this paper are summarised in the following two theorems. For a given set of risk aversion parameter $R$, discount factor $\delta$ and market parameters $r$, $\mu$, $\sigma$, $\alpha$, $\eta$, $\rho$, we say the problem is ({\cblue unconditionally}) well-posed if the value function is finite on the interior of the solvency region for all values of the transaction costs {\cblue $\lambda \geq 0$ and $\gamma \in [0,1)$ with $\lambda + \gamma>0$}. We say the problem is ill-posed if the value function is infinite for all $\lambda$ and $\gamma$. We say the problem is conditionally well-posed if the problem is well-posed for large values of the round-trip transaction cost, but ill-posed for small values. Theorems~\ref{thm:wellposed} and \ref{thm:valfun} are proved in Appendix~\ref{app:pfmain}.

\begin{thm}
The investment/consumption problem is:
\begin{enumerate}
\item well-posed in either of the following cases:
	\begin{enumerate}
		\item $R>1$,
		\item $R<1$ and $m_{M} \geq 0$;
	\end{enumerate}
\item ill-posed if $R<1$, $m_{M}<0$ and $\ell(1)\leqslant0$;
\item conditionally well-posed if $R<1$, $m_{M}<0$ and $\ell(1)>0$. In this case the problem is well-posed if and only if $\xi > \overline{\xi}$
where $\overline{\xi}$ is defined in \eqref{eq:critical_trancost} below.
\end{enumerate}
\label{thm:wellposed}
\end{thm}

{\cblue Note that, if $R<1$ then $m_M>0$ is necessary and sufficient for the problem with transaction  costs set to zero to be well-posed.
Further, if $R<1$ and $m_M=0$ and $\lambda = 0 = \gamma$ (a case we have excluded) then the problem is ill-posed for zero transaction costs, but well-posed for non-zero transaction costs.

The following result follows from the proof of Theorem~\ref{thm:wellposed} in the ill-posed case and relies on the fact that in this case there is an admissible strategy which generates infinite expected utility without investing in the liquid risky asset $S$.

\begin{cor}
\label{cor:wellposed}
The problem with one risky liquid asset and one illiquid risky asset is ill-posed (for all values of transaction costs) if and only if the problem with the
risky liquid asset omitted is ill-posed (for all values of transaction costs).
\end{cor}
}

\begin{thm}
Suppose the parameters are such that the problem is well-posed. Set
\begin{align*}
V^{C}(x,y,\theta)=\left(\frac{b_{1}}{Rb_{4}}\right)^{-R}\frac{(x+y\theta)^{1-R}}{1-R}G^{C}\left(\frac{y\theta}{x+y\theta}\right)
\end{align*}
where $G^{C}$ is as defined as in the relevant case of Proposition \ref{prop:def_g}. Then $V^{C}=V$ where $V$ is the value function of the investment/consumption problem defined in \eqref{eq:thevalfun}.
\label{thm:valfun}
\end{thm}

\section{Solutions to the free boundary value problem}
\label{sect:fbp}

Let $\mathcal{S} \subseteq \{(q,n); q>0, n \geq 0 \}$ be the set $\mathcal{S} = \{ q=1 \} \cup \{ q = \frac{R}{R-1} \} \cup \{ n = 0 \} \cup \{ q < 1, (1-R)n \geq (1-R)\ell(q) \}$.
On $(0,\infty) \times [0,\infty) \setminus \mathcal{S}$ define $F(q,n) = O(q,n)/n$. Extend the definition of $F$ to $(0,\infty) \times [0,\infty)$ where possible by taking appropriate limits.
We begin this section with a list of useful results regarding the functions $m$ and $\ell$ and operators $O$ and $F$. The proof of Lemma~\ref{lemma:list} is given in Appendix~\ref{app:ode}.
\begin{lemma}
\begin{enumerate}
	\item \begin{enumerate}
			\item For $R<1$, $\ell(q)>m(q)$ on $q\in(0,1]$. Moreover, on $(0,\infty)$, $m$ crosses $\ell$ exactly once from below at some point above $1$;
			\item For $R>1$, $m(q)>\ell(q)$ on $q\in(0,1]$. Moreover, on $(0,\infty)$, $m$ either does not cross $\ell$ at all, or touches $\ell$ exactly once in the open interval $(1,R/(R-1))$, or crosses $\ell$ twice on $(1,R/(R-1))$.
		\end{enumerate}
	\item 
	For $R>1$, $F(q,n)$ is well defined at $q=R/(R-1)$.
	\item For $n>0$ and $(1-R)n<(1-R)\ell(1)$, $F(1,n)$ is well-defined and
	\begin{align}
		F(1,n) := \lim_{q\to 1}F(q,n)=-\frac{(1-R)(n-m(1))}{\ell(1)-n}.
		\label{eq:limatone}
	\end{align}
	Also, for $q\leq1$ and $R<1$ we have $\lim_{n\uparrow \ell(q)}F(q,n)=-\infty$ (and $\lim_{n\downarrow \ell(q)}F(q,n)=+\infty$ if $R>1$). For $q>1$ and $R<1$ (and $1<q<\frac{R}{R-1}$ for $R>1$) we have
	\begin{align}
		F(q, \ell(q)) :=\lim_{n\to \ell(q)}F(q,n)=-\frac{1-R}{R(1-q)}\left\{\frac{q[(1-R)q+R]}{[(1-R)q+R]^{2}+(b_{2}-1)R^{2}}-1\right\}.
		\label{eq:limatn}
	\end{align}
\item $F(q,n)=0$ if and only if $n=m(q)$. Moreover,
	\begin{enumerate}
		\item for $R<1$:
			\begin{enumerate}
				\item On $0<q<1$, $F(q,n)<0$ for $m(q)<n<\ell(q)$ and $F(q,n)>0$ for $n<m(q)$ or $n>\ell(q)$;
				\item At $q=1$, $F(1,n)<0$ for $m(1)<n<\ell(1)$ and $F(1,n)>0$ for $n<m(1)$. $F(1,n)$ is not well-defined for $n\geqslant \ell(1)$;
				\item On $q>1$, $F(q,n)<0$ for $n>m(q)$ and $F(q,n)>0$ for $n<m(q)$;
			\end{enumerate}
		\item for $R>1$:
		\begin{enumerate}
				\item On $0<q<1$, $F(q,n)>0$ for $\ell(q)<n<m(q)$ and $F(q,n)<0$ for $n<\ell(q)$ or $n>m(q)$;
				\item At $q=1$, $F(1,n)>0$ for $\ell(1)<n<m(1)$ and $F(1,n)<0$ for $n>m(1)$. $F(1,n)$ is not well-defined for $n\leqslant \ell(1)$;
				\item On $1<q\leqslant R/(R-1)$, $F(q,n)<0$ for $n>m(q)$ and $F(q,n)>0$ for $n<m(q)$;
				\item On $q>R/(R-1)$, $F(q,n)<0$ for $m(q)<n<\ell(q)$ and $F(q,n)>0$ for $n>\ell(q)$ or $n<m(q)$.
			\end{enumerate}
	\end{enumerate}
\end{enumerate}
\label{lemma:list}
\end{lemma}


Recall $(q_{M},m_{M})$ is the extreme point of the quadratic $m$ (a minimum when $R<1$ and a maximum when $R>1$) with $q_{M}=\frac{b_{3}}{2R}>0$. The key analytical properties of the problem only depend on the signs of the three parameters $(1-R,m_{M},\ell(1))$. We classify four different cases using the decision tree in Figure \ref{fig:case}.

\begin{figure}
\centering
\begin{tikzpicture}[level 1/.style={sibling distance=50mm},level 2/.style={sibling distance=40mm}]
  \node {} 
    child {node {$R<1$}
    		child {node {$m_{M}>0$}
		child {node {Case 1 (W)}}}
		child {node {$m_{M}<0$}
		child {node {$\ell(1)\leqslant 0$}
		child {node {Case 2 (I)}}}
		child{node {$\ell(1)>0$}
		child {node {Case 3 (CW)}}}
		}
		} 
    child {node {$R>1$} 
      		child {node {Case 4 (W)}}  
    		}; 
\end{tikzpicture}
\caption{Classification of different cases based on the signs of the parameters. The abbreviations in parentheses indicate the solution features of the cases, where ``W" refers to unconditional well-posedness for all levels of transaction cost, ``I" refers to unconditional ill-posedness for all levels of transaction cost and ``CW" refers to conditional well-posedness, i.e. well-posedness for sufficiently high levels of transaction cost only.}
\label{fig:case}
\end{figure}
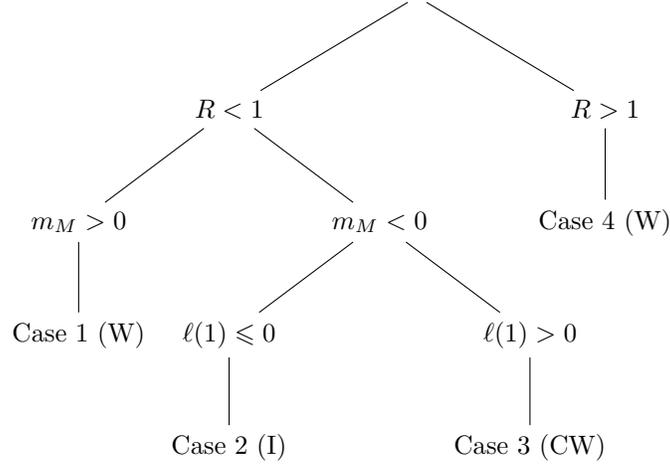


We parameterise the family of solutions to \eqref{eq:freebound} by the left boundary point. Fix $u$ and denote by $(n_{u}(q))_{q\geqslant u}$ the solution to the initial value problem
\begin{align*}
n'(q)=O(q,n(q)),\quad n(u)=m(u).
\end{align*}
Let $\zeta(u)=\inf \{ q\geqslant u:(1-R)n_{u}(q)<(1-R)m(q) \}$ denote where $n_{u}$ first crosses $m$ to the right of $u$.
Define
\begin{align}
\Sigma(u) = \exp\left(\int_{u}^{\zeta(u)}\left(-\frac{R}{q(1-R)}\frac{O(q,n_{u}(q))}{n_{u}(q)}\right)dq\right)-1.
\label{eq:trancost2}
\end{align}

\begin{lemma}  Suppose $m_M>0$. Then
$\Sigma$ is a strictly decreasing, continuous mapping $\Sigma:(0,q_{M}]\to[0,\infty)$ with $\Sigma(0+)=+\infty$ and $\Sigma(q_{M})=0$.

Now suppose $m_M \leq 0$. Let $p_- \leq p_+$ be the roots of $m(q)=0$. Set
\begin{equation}
 \overline{\xi}:=\lim_{u\uparrow p_{-}}\Sigma(u)=\exp\left(-\int_{p_{-}}^{p_{+}}\frac{R}{q(1-R)}F(q,0)dq\right)-1.
 \label{eq:critical_trancost}
\end{equation}
Then $\Sigma$ is a strictly decreasing, continuous mapping $\Sigma:(0,p_-]\to[\overline{\xi},\infty)$ with $\Sigma(0+)=+\infty$ and $\Sigma(p_-)=\overline{\xi}$. Moreover, $\lim_{u\uparrow p_{-}}n_{u}(\cdot)=0$ and $\lim_{u\uparrow p_{-}}\zeta(u)=p_{+}$.
\label{lemma:onto}
\end{lemma}
Lemma~\ref{lemma:onto} is proved in Appendix~\ref{app:ode}.

\subsection{The cases}

\subsubsection{Case 1: $R<1$ and $m_M \geq 0$}
For any initial value $u\in(0,q_{M})$, $m'(u)<0=O(u,m(u))=O(u,n_{u}(u))=n_{u}'(u)$. Thus $n_{u}(q)$ must initially be larger than $m(q)$ for $q$ being close to $u$. By part 4 of Lemma \ref{lemma:list}, $O(q,n)$ is negative on $\{(q,n):0<q\leqslant 1,m(q)<n<\ell(q)\}\cup\{(q,n):q>1,n>m(q)\}$. Also, $n_{u}(q)$ cannot cross $l(q)$ from below on $0<q\leqslant 1$ since $\lim_{n\uparrow \ell(q)}O(q,n)=-\infty$.  By considering the sign of $O(q,n)$, we conclude $n_{u}$ must be decreasing until it crosses $m$. This guarantees the finiteness of $\zeta(u)$, and the triple $(n_{u}(\cdot),u,\zeta(u))$ represents one possible solution to problem \eqref{eq:freebound}. Notice that the family of solutions $(n_u(\cdot))_{0<u<q_M}$ cannot cross, and thus $n_{u}(q)$ is decreasing in $u$. The solutions corresponding to initial values $u=0$ and $u=q_{M}$ can be understood as the appropriate limit of a sequence of solutions.

Although $O(q,n)$ has singularities at $q=1$ and $n=\ell(q)$, part 3 of Lemma \ref{lemma:list} shows that a well-defined limit $O(q,n)$ exists on $\{(q,n):q=1,n<\ell(1)\}$ and $\{(q,n):q>1,n=\ell(q)\}$. Hence there exists a continuous modification of $O(q,n)$ and a solution $n_{u}$ can actually pass through these singularity curves. See Figure \ref{fig:case1}(a) for some examples. 

From the analysis leading to \eqref{eq:trancost}, the correct choice of $u$ should satisfy $\xi = \Sigma(u)$.
From Lemma \ref{lemma:onto}, for every given level of round-trip transaction cost $\xi$, there exists a unique choice of the left boundary point given by $u_{*}=\Sigma^{-1}(\xi)$ and then the desired solution to the free boundary value problem is given by $(n_{u_{*}}(\cdot),u_{*},\zeta(u_{*}))$. Figure \ref{fig:case1}(b) gives the plots of $\Sigma^{-1}(\xi)$ and $\zeta(\Sigma^{-1}(\xi))$ representing the boundaries $(q_{*},q^{*})$ under different levels of transaction cost.

\begin{figure}[!htbp]
	\captionsetup[subfigure]{width=0.4\textwidth}
	\centering
	\subcaptionbox{Examples of solutions $n_{u}(q)$ with different initial values $(u,m(u))$.}{\includegraphics[scale =0.35] {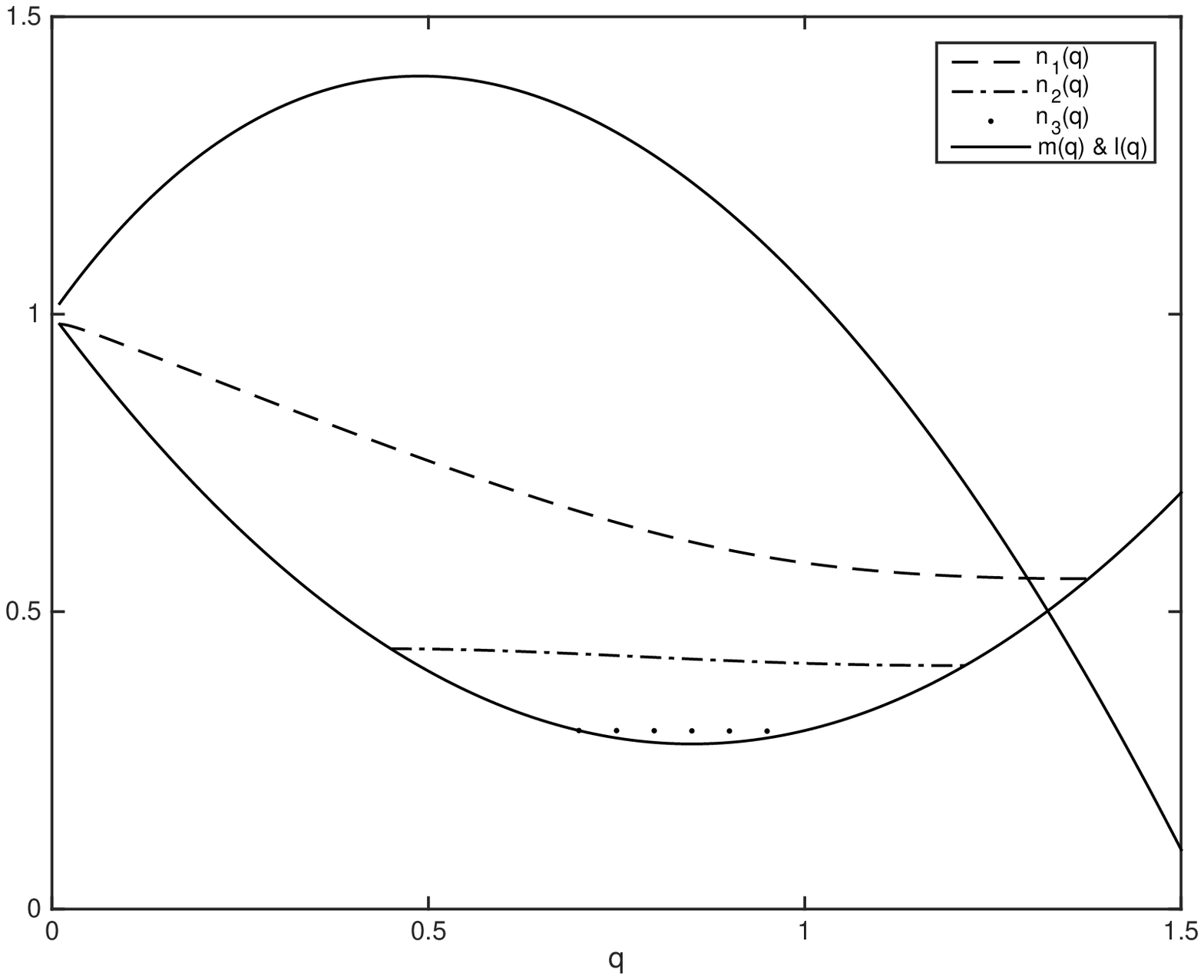}}
	\subcaptionbox{Plots of $q_{*}=\Sigma^{-1}(\xi)$ and $q^{*}=\zeta(q_{*})$.}{\includegraphics[scale =0.35]{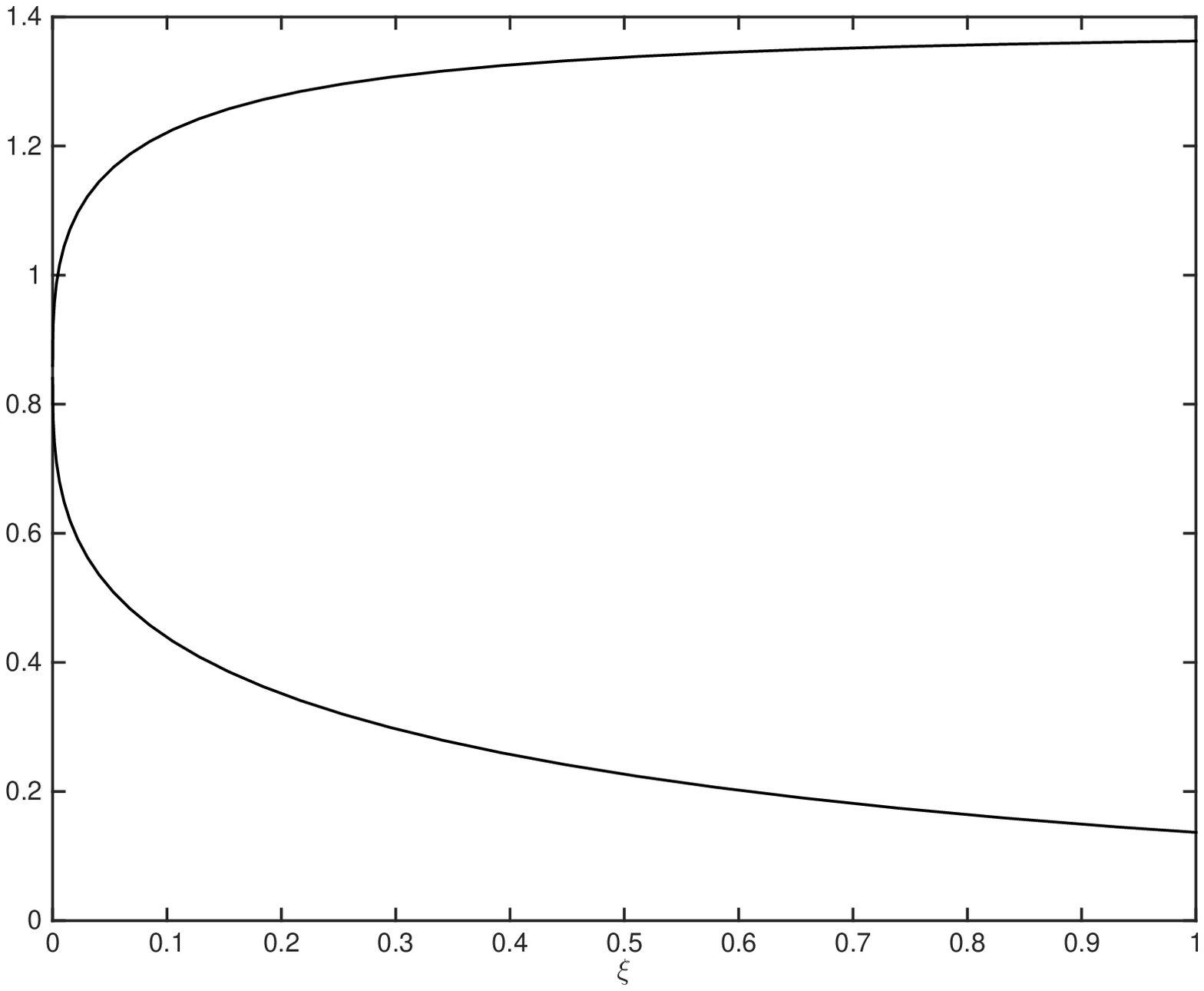}}
	
\caption{Case 1 where parameters chosen are $R=0.5$, $b_{1}=0.25$, $b_{2}=1.75$ and $b_{3}=0.85$.}
\label{fig:case1}
\end{figure}

{\cblue Based on Figure~\ref{fig:case1} we can make a series of simple observations about the behaviour of $q_*$ and $q^*$ (which hold in the other cases too) some of which will be proved in Section~\ref{sec:compstat} on the comparative statics of the problem. First, the lower and upper boundaries of the no transaction region, expressed in terms $q_*$ and $q^*$, are monotonic decreasing and monotonic increasing respectively. In particular, the no-transaction region gets wider as transaction costs increase. Second, the no-transaction region may be contained in the first quadrant $(0 < q_* < q^* < 1)$, or the upper-half plane $(0 < q_* < 1 < q^*)$, depending on $\xi$ and for other parameter values we may have that the no-transaction region is contained in the second quadrant ($1 < q_* < q^*$). Third, $lim_{\xi \downarrow 0} q_* = q_M = \lim_{\xi \downarrow 0} q^*$. Moreover, the numerics are suggestive of $\lim_{\xi \uparrow \infty} q_* = 0$ and $\lim_{\xi \uparrow \infty} q^* =: q^*_\infty < \infty$ so that there is a part of the solvency space close to the solvency limit $p = 1/\gamma$ which, even in the regime of very large transaction costs, is inside the region where a sale of $Y$ at $t=0$ is necessary. Fourth, $q^*$ is less sensitive to changes in $\xi$ than $q^*$ so that the no-transaction wedge is not centred on the Merton line.}

\subsubsection{Case 2: $R<1$, $m_M<0$, $\ell(1)\leq0$}

Let $\ell_{0}$ be the root of $\ell(q)=0$ on $q\in(0,1)$. Since the solution of $n'(q)=O(q,n(q))$ must be bounded below by zero and above by $\ell(q)$ for $q\in(0,\ell_0)$, for any initial value $(u,m(u))$ for which $m(u)>0$, the corresponding solution $n_{u}(\cdot)$ must hit $(\ell_{0},0)$. Hence there does not exist any positive solution which crosses $m$ again to the right of $u$. See Figure \ref{fig:case2}. In this case, there is no solution to the free boundary value problem and indeed the underlying problem is ill-posed for all levels of transaction costs and thus the value function cannot be defined.

\begin{figure}[!htbp]
\begin{center}
\includegraphics[scale=0.35]{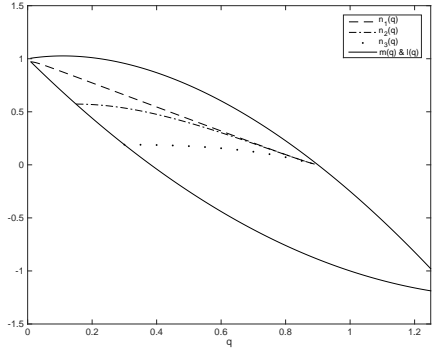}
\caption{Case 2 where parameters chosen are $R=0.5$, $b_{1}=0.25$, $b_{2}=1.75$ and $b_{3}=1.5$. {\cblue Then $q_M = \frac{b_3}{2R}=0.85$.}}
\label{fig:case2}
\end{center}
\end{figure}

\subsubsection{Case 3: $R<1$, $m_M<0$, $\ell(1) > 0$}

Let $p_{\pm}$ with $0<p_{-}<q_{M}<p_{+}$ be the two roots of $m(q)=0$. The parameterisation of the solution is the same as in Case 1 except the left boundary point should now be restricted to $u\in(0,p_{-})$ to ensure a positive initial value. The function $\Sigma$ defined in \eqref{eq:trancost2} is still a strictly decreasing map with $\Sigma(0+)=+\infty$ except its domain is now restricted to $(0,p_{-}]$.

Unlike Case 1, we now only consider $\Sigma^{-1}(\xi)$ on the range $\xi\in(\overline{\xi},\infty)$. For such a given high level of round-trip transaction cost, the required left boundary point is given by $u_{*}=\Sigma^{-1}(\xi)$ and $u^{*}=\zeta(u_{*})$, see Figure \ref{fig:case3}. In this case, the problem is conditionally well-posed, {\cblue ie it is well-posed} only for a sufficiently high level of transaction cost.

\begin{figure}[!htbp]
	\captionsetup[subfigure]{width=0.4\textwidth}
	\centering
	\subcaptionbox{Examples of solutions $n_{u}(q)$ with different initial values $(u,m(u))$.}{\includegraphics[scale =0.35] {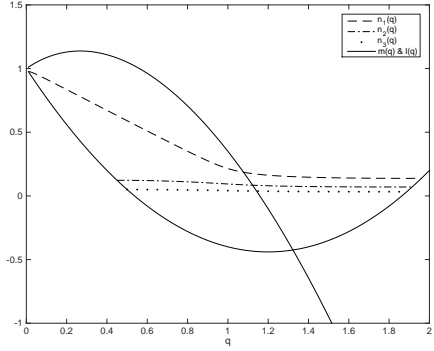}}
	\subcaptionbox{Plots of $q_{*}=\Sigma^{-1}(\xi)$ and $q^{*}=\zeta(q_{*})$. {\cblue $q_*$ and $q^*$ are not defined for $\xi < \bar{\xi}$.}}{\includegraphics[scale =0.35]{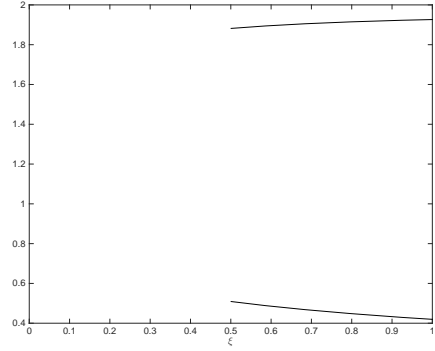}}
	
\caption{Case 3 where parameters chosen are $R=0.5$, $b_{1}=0.25$, $b_{2}=1.75$ and $b_{3}=1.2$.}
\label{fig:case3}
\end{figure}

\subsubsection{Case 4: $R>1$.}

In this case the quadratic $m$ has a positive maxima at $(q_{M},m_{M})$ and $m(q)>\ell(q)$ on $q\in(0,1)$. By checking the sign of $O(q,n)$ using part 4 of Lemma \ref{lemma:list}, one can verify that the solution $n_{u}$ of the initial value problem is always increasing for any choice of left boundary point $u\in(0,q_{M})$. In this case the family of solutions is increasing in $u$. The solution $n_{u}(q)$ crosses $m(q)$ from below at $\zeta(u)=\inf(q\geqslant u:n_{u}(q)>m(q))$. The correct choice of $u$ is again the one solving $\xi=\Sigma(u)$ using the same definition in \eqref{eq:trancost2}. As in Case 1, the function $\Sigma$ is onto from $(0,q_{M}]$ to $[0,\infty)$ and hence $u_{*}=\Sigma^{-1}(\xi)$ always exists uniquely for any $\xi$. See Figure \ref{fig:case4}. Indeed for $R>1$, the agent's utility function is always bounded above by zero and hence the value function always exists and is finite.

\begin{figure}[!htbp]
	\captionsetup[subfigure]{width=0.4\textwidth}
	\centering
	\subcaptionbox{Examples of solutions $n_{u}(q)$ with different initial values $(u,m(u))$.}{\includegraphics[scale =0.35] {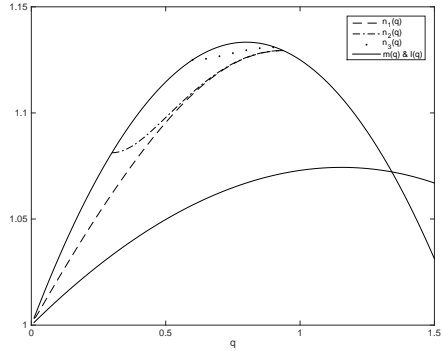}}
	\subcaptionbox{Plots of $q_{*}=\Sigma^{-1}(\xi)$ and $q^{*}=\zeta(q_{*})$.}{\includegraphics[scale =0.35]{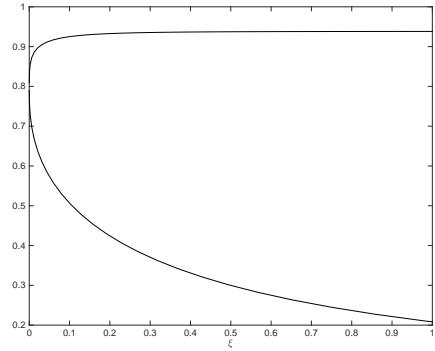}}
	
\caption{Case 4 where parameters chosen are $R=1.25$, $b_{1}=1.5$, $b_{2}=1.25$ and $b_{3}=2$.}
\label{fig:case4}
\end{figure}

\section{Comparative statics}
\label{sec:compstat}

In this section, we investigate how the no-transaction wedge $[p_{*},p^{*}]$ and the value function $V$ change with the market parameters and level of transaction costs.

\subsection{Monotonicity with respect to market parameters}


\begin{prop}
Suppose $(n(\cdot),q_{*},q^{*})$ is the solution to the free boundary value problem. Then:
\begin{enumerate}
	\item $q_{*}$ and $q^{*}$ are decreasing in $b_{1}$; 
	\item For $R<1$, $q_{*}$ and $q^{*}$ are increasing in $b_{3}$.
\end{enumerate}
\label{prop:compstat}
\end{prop}

Proposition~\ref{prop:compstat} is proved in Appendix~\ref{app:compstat}.

Recall that
$p_{*}=\frac{q_{*}}{1+\lambda(1-q_{*})}$ and  $p^{*}=\frac{q^{*}}{1-\gamma(1-q^{*})}$.
Then, Proposition~\ref{prop:compstat} gives immediately:
\begin{thm}
\begin{enumerate}
	\item $p_{*}$ and $p^{*}$ are decreasing in $b_{1}$; 
	\item For $R<1$, $p_{*}$ and $p^{*}$ are increasing in $b_{3}$.
\end{enumerate}
\label{thm:compstat}
\end{thm}

Theorem~\ref{thm:compstat} describes the comparative statics in terms of the auxiliary parameters.\footnote{Since the free boundary value problem does not depend on $b_4$, $q_*$ and $q^*$ are trivially independent of $b_4$. We have strong numerical evidence that $q_*$ is decreasing in $b_{2}$ and $q^{*}$ is increasing in $b_{2}$, but we have not been able to prove this result.} In general, it is difficult to make categorical statements about the comparative statics with respect to the original market parameters since many of the market parameters enter the definitions of more than one of the auxiliary parameters. However, we have the following results concerning the dependence of $p_*$ and $p^*$ on the discount rate, and on the drift of the illiquid asset.

\begin{cor} $p_{*}$ and $p^{*}$ are decreasing in $\delta$. If $R<1$ then $p_{*}$ and $p^{*}$ are increasing in $\alpha$.
\label{cor:compstat}
\end{cor}

{\cblue The corollary confirms the intuition that as the return on the illiquid asset asset increases, it becomes more valuable and the agent elects to buy the illiquid asset sooner, and to sell it later. Moreover, as his discount parameter increases, he wants to consume wealth sooner, and since consumption takes place from the cash account he elects to keep more of his wealth in liquid assets, and less in the illiquid asset.}

Now we consider the cash value of the holdings in the illiquid asset. We compare the agent with holdings in the illiquid asset to an otherwise identical agent (same risk aversion and discount parameter, and trading in the financial market with bond and risky asset with price $S$) who has a zero initial endowment in the illiquid asset and is precluded from taking any positions in the risky asset.

Consider the market without the illiquid asset. For an agent operating in this market a consumption/investment strategy is admissible for initial wealth $x>0$ (we write $(C = (C_t)_{t \geq 0},\Pi = (\Pi_t)_{t \geq 0}) \in \mathcal{A}_W(x)$) if $C$ and $\Pi$ are progressively measurable, and if the resulting wealth process $X=(X_t)_{t\geq0}$ is non-negative for all $t$. Here $X$ solves
\[ dX_t = r(X_t - \Pi_t) dt + \frac{\Pi_t}{S_t} dS_t - C_t dt \]
subject to $X_0=x$. Let $W=W(x)$ be the value function for a CRRA investor:
\[ W(x) = \sup_{(C,\Pi) \in \mathcal{A}_W(x)} {\mathbb E} \left[ \int_0^\infty e^{-\delta t} \frac{C_t^{1-R}}{1-R} dt \right] .\]
The problem of finding $W$ is a classical Merton consumption/investment problem without transaction costs. We find
\[ W(x) = \left[ \frac{1}{R} \left( \delta - r(1-R) - \frac{\beta^2(1-R)}{2R} \right) \right]^{-R} \frac{x^{1-R}}{1-R} = \left( \frac{b_1}{b_4 R} \right)^{-R} \frac{x^{1-R}}{1-R} .\]

Define $\mathcal{C} = \mathcal{C}(y\theta;x)$ to be the certainty equivalent value of the holding of the illiquid asset, i.e. the cash amount which the agent with liquid wealth $x$ and $\theta$ units of the illiquid asset with current price $y$, trading in the market with transaction costs, would exchange for his holdings of the illiquid asset, if after this exchange he is not allowed to trade in the illiquid asset. (We assume there are no transaction costs on this exchange, but they can be easily added if required.) Then $ \mathcal{C} = \mathcal{C}(y\theta;x)$ solves
\[ W(x + \mathcal{C}) = V(x,y,\theta) \]
which becomes
\[ \mathcal{C} = \mathcal{C}(y\theta;x) = (x + y \theta) G(p)^{1/(1-R)} - x. \]
Theorem~\ref{thm:compstat2} is proved in Appendix~\ref{app:compstat}.

\begin{thm}
\begin{enumerate}
    \item $(1-R)G$ is decreasing in $b_{1}$;
    \item $(1-R)G$ is increasing in $b_{3}$.
\end{enumerate}
\label{thm:compstat2}
\end{thm}

\begin{cor}
$\mathcal{C}$ is decreasing in $\delta$ and increasing in $\alpha$.
\label{cor:compstat2}
\end{cor}

Both these monotonicities are intuitively natural.
For the monotonicity in $\alpha$,  since the agent only ever holds long\footnote{Note, if he starts with a solvent initial portfolio, but with a negative holding in the illiquid asset, then the agent makes an instantaneous transaction at time zero to make his holding positive.} positions in the illiquid asset, we expect him to benefit from an increase in drift and hence price of the illiquid asset. (Note, some care is needed in making this argument precise. Part of the optimal strategy is to sometimes purchase units of the illiquid asset, and this will be more costly if the price is higher.)
If we consider monotonicty in $\delta$ then for $R<1$, increasing $\delta$ reduces the magnitude of the discounted utility of consumption, and reduces the value function. However, this is not the same as decreasing the certainty equivalent value of the holding of risky asset. Indeed, when $R>1$, increasing $\delta$ reduces the magnitude of the discounted utility of consumption, but since the terms are negative, this increases the value function. Nonetheless, $\mathcal{C}$ is decreasing in $\delta$.

\subsection{Monotonicity with respect to transaction costs}

From the discussion in Section \ref{sect:fbp}, we have seen that transformed boundaries only depends on the round-trip transaction cost $\xi$. In particular, $q_{*}$ and $q^{*}$ are respectively strictly decreasing and increasing in $\xi$. However, the purchase/sale boundaries in the original scale still depend on the individual costs of purchase and sale. 
Write
\begin{align*}
p_{*}(\lambda,\gamma)=\frac{q_{*}(\xi)}{1+\lambda(1-q_{*}(\xi))},\quad  p^{*}(\lambda,\gamma)=\frac{q^{*}(\xi)}{1-\gamma(1-q^{*}(\xi))}
\end{align*}
and recall that $\xi=\frac{\lambda+\gamma}{1-\gamma}$.
{\cblue If $q_*(\xi) < q^*(\xi) < 1$ then $p_*(\lambda, \gamma) \leq q_*(\xi) < q_M < q^*(\xi) \leq p^*(\lambda, \gamma)$ and the Merton line lies inside the no-transaction wedge. However, if $1 < q_*(\xi) < q^*(\xi)$ then we have $p_*(\lambda, \gamma) > q_*(\xi)$ and $p^*(\lambda, \gamma) < q^*$ and the Merton line may fall outside the no-transaction region.} We have
\begin{align*}
\frac{dp_{*}}{d\gamma}=\frac{\partial p_{*}}{\partial q_{*}}\frac{\partial q_{*}}{\partial \xi}\frac{\partial \xi}{\partial \gamma}
=\frac{1+\lambda}{(1-\gamma)^{2}}\frac{1+\lambda}{[1+\lambda(1-q_{*})]^{2}}\frac{\partial q_{*}}{\partial \xi}<0
\end{align*}
so that the critical ratio of wealth in the illiquid asset to paper wealth at which the agent purchases more illiquid asset is decreasing in the transaction cost on sales. However, perhaps surprisingly, the dependence of the critical ratio $p_*$ at which purchases occur on the transaction cost on purchases is not unambiguous in sign:
\begin{align*}
\frac{dp_{*}}{d\lambda}=\frac{\partial p_{*}}{\partial \lambda}+\frac{\partial p_{*}}{\partial q_{*}}\frac{\partial q_{*}}{\partial \xi}\frac{\partial \xi}{\partial \lambda}
=-\frac{q_{*}(1-q_{*})}{[1+\lambda(1-q_{*})]^{2}}+\frac{1}{1-\gamma}\frac{1+\lambda}{[1+\lambda(1-q_{*})]^{2}}\frac{\partial q_{*}}{\partial \xi}
\end{align*}
is not necessarily negative, for we may have $q_*>1$.
This issues are discussed further in Hobson et al~\cite{HobsonTseZhu:16} where examples are given {\cblue in which the Merton line lies outside the no transaction region and} in which the boundaries to the no-transaction region are not monotonic in the transaction cost parameters.

\section{Conclusion}
\label{sec:conc}

Merton's solution~\cite{Merton:69} of the infinite horizon, consumption and investment problem is elegant and insightful but assumes a perfect market with no frictions. Building on this work, there is a large literature, starting with Constantinides and Magill~\cite{ConstantinidesMagill:76} and Davis and Norman~\cite{DavisNorman:90} investigating the form of the solution in the presence of transaction costs. When there is a single asset Choi et al~\cite{ChoiSirbuZitkovic:13} (via shadow prices) and Hobson et al~\cite{HobsonTseZhu:16} (via an analysis of the HJB equation) are able to characterise precisely when the problem is well-posed. However, \cite{DavisNorman:90}, \cite{ChoiSirbuZitkovic:13}
and \cite{HobsonTseZhu:16} all assume the financial market includes just a single risky asset.

In this paper we have extended the results to two risky assets, and give a complete characterisation of the solution, but in the special case where transaction costs are payable on only one of the risky assets. This is also the model studied by Choi~\cite{Choi:16} using different methods. The presence of the second risky asset, which may be used for hedging and investment purposes, makes the problem significantly more complicated than the single risky asset case, but we can extend the methods of \cite{HobsonTseZhu:16} to give a complete solution. Indeed, up to evaluating an integral of a known algebraic function, we can determine exactly when the problem is well-posed and up to solving a free boundary value problem for a first order differential equation we can determine the boundaries of the no-transaction wedge.

At the heart of our analysis is this free boundary value problem. Although the utility maximisation problem depends on many parameters describing the agent (his risk aversion and discount rate), the market (the interest rate and the drifts, volatilities and correlations of the traded assets) and the frictions (the transaction costs on sales and purchases) the ODE depends on the risk aversion parameter and just three further parameters, and the solution we want can be specified further in terms of the round-trip transaction cost.

Building on the work of Choi et al~\cite{ChoiSirbuZitkovic:13}, in our previous work~\cite{HobsonTseZhu:16} we give a solution to the problem in the case of a single risky asset. The major issue in \cite{ChoiSirbuZitkovic:13} and \cite{HobsonTseZhu:16} is to understand the solution of an ODE as it passes through a singular point. In this paper the problem is richer, and the ODE is more complicated, but in other ways the analysis is much simpler because although the key ODE has singularities, these can be removed.

In the paper we have assumed a single illiquid asset and just one further risky asset, but the analysis extends immediately to the case of a single illiquid asset and several risky assets on which no transaction costs are payable, at the expense of a more complicated notation. This observation is a form of mutual fund theorem --- the agent chooses to invest in the additional liquid financial assets in fixed proportions and these assets may be combined into a representative market asset. Details of the argument in a related context may be found in Evans et al~\cite{EvansHendersonHobson:08}. Nonetheless, the extension to a model with many risky assets with transaction costs payable on all of them remains a challenging open problem.

\bibliographystyle{plain}


\appendix

\section{Transformation of the HJB equation}
\label{subsect:hjb}
Looking at the HJB Equation \eqref{eq:hjb}, and using intuition gained from similar problems, we expect that $V=V(x,y,\theta)$ can be written as $V(x,y,\theta)=x^{1-R}J(\frac{y \theta}{x})$ for $J$ a function of a single variable $z$ representing the ratio of wealth in the illiquid asset to wealth in the liquid assets. The equation for $J=J(z)$ contains expressions of the form $zJ'(z)$ and $z^2 J''(z)$ and so can be made into a homogeneous equation by the substitution $(z,J(z)) \mapsto (e^u,K(u))$. The second-order equation for $K$ can then be reduced to a first order equation by setting $w(K) = \frac{dK}{du}$ and making $K$ the subject of the equation, see \cite{EvansHendersonHobson:08} or \cite{HobsonTseZhu:16} for details of a similar order-reduction in a related problem. However, there are cases where $x=0$ lies inside the no-transaction region and at this point $z$ is undefined, and the above approach does not work. Hence, we need to use a different parametrisation. We use a parametrisation based on $P_t = \frac{Y_t \Theta_t}{X_t + Y_t \Theta_t}$ representing the proportion of paper wealth which is held in the illiquid asset. The delicate point at $x= \pm 0$ (or $z = \pm \infty$) becomes a delicate point at $p=1$, but as we show by a careful analysis any singularities can be removed.

Using the form of value function in \eqref{eq:valfun} to compute all the relevant partial derivatives, \eqref{eq:hjb} can be rewritten as
\begin{align}
0&=\frac{b_{1}}{b_{4}}\left[G(p)-\frac{pG'(p)}{1-R}\right]^{1-1/R}-\delta G(p)+r(1-p)\left[(1-R)G(p)-pG'(p)\right] \nonumber \\
&+\alpha \left[(1-R)pG(p)+p(1-p)G'(p)\right]  +\frac{\eta^{2}}{2}\left[p^{2}(1-p)^{2}G''(p)-2Rp^{2}(1-p)G'(p)-R(1-R)p^{2}G(p)\right] \nonumber \\
&-\frac{\left\{\beta \left[(1-R)G(p)-pG'(p)\right]+\eta\rho \left[-R(1-R)pG(p)+Rp(2p-1)G'(p)-p^{2}(1-p)G''(p)\right]\right\}^{2}}{2\left[p^{2}G''(p)+2RpG'(p)-R(1-R)G(p)\right]}.
\label{eq:hjb2}
\end{align}

Let\footnote{The assumption $b_3>0$ means that the agent would like to hold positive quantities of the illiquid asset, and that the no-transaction wedge is contained in the half-space $p>0$. To allow for $b_3<0$ it is necessary to consider $p<0$. This case can be incorporated into the analysis by incorporating an extra factor of $\sgn(p)$ into the definition of $h$, so that $h(p)=\sgn(p(1-p))|1-p|^{R-1}G(p)$. This then leads to extra cases, but no new mathematics, and the problem can still be reduced to solving $n'=O(q,n)$ where $O$ is given by \eqref{eq:formO}, but now for $q<0$.}
$h(p)=\sgn(1-p)|1-p|^{R-1}G(p)$ and $w(h)=p(1-p)\frac{dh}{dp}$. Then
\begin{equation}
\frac{w(h)}{p(1-p)}=\frac{dh}{dp}=\sgn(1-p)|1-p|^{R-1}\left[G'(p)+(1-R)\frac{G(p)}{1-p}\right]
\label{eq:dhdp}
\end{equation}
and in turn
\begin{align}
G'(p)=\frac{w(h)}{|p||1-p|^{R}}-(1-R)\frac{G(p)}{1-p}.
\end{align}
This gives 
\begin{eqnarray}
G(p)-\frac{pG'(p)}{1-R}&=&G(p)-\frac{p}{1-R}\left[\frac{w(h)}{|p||1-p|^{R}}-(1-R)\frac{G(p)}{1-p}\right] \nonumber \\
&=&|1-p|^{-R}h\left(1-\frac{w(h)}{(1-R)h}\right). \label{eq:Ghw}
\end{eqnarray}
We expect that $V_x > 0$ and hence that this expression is positive. It follows that $\sgn(1-p) = \sgn(h) = \sgn(1 - W(h))$.
Then
\begin{eqnarray}
\left(G(p)-\frac{pG'(p)}{1-R}\right)^{1-1/R} 
&=&\sgn(1-p)|1-p|^{1-R}h|h|^{-1/R}\left|1-\frac{w(h)}{(1-R)h}\right|^{1-1/R}, \label{eq:Vxdef} \\ 
(1-p)\left[(1-R)G(p)-pG'(p)\right]&= &
\sgn(1-p)|1-p|^{1-R}\left[(1-R)h-w(h)\right], \nonumber \\
(1-R)G(p)-pG'(p) 
&=&\frac{\sgn(1-p)|1-p|^{1-R}}{1-p}(1-R)h\left(1-\frac{w(h)}{(1-R)h}\right) \nonumber
\end{eqnarray}
and
\begin{eqnarray*}
(1-R)pG(p)+p(1-p)G'(p)
&=&\sgn(1-p)|1-p|^{1-R}w(h).
\end{eqnarray*}

Taking a further derivative
\begin{align*}
w(h)w'(h)&=p(1-p)\frac{dh}{dp}\frac{d}{dh}w(h)=p(1-p)\frac{d}{dp}w(h) \\
&=p(1-p)\frac{d}{dp}\left\{\sgn(1-p)|1-p|^{R-1}\left[p(1-p)G'(p)+(1-R)pG(p)\right] \right\} \\
&=\sgn(1-p)|1-p|^{R-1}\left[p^{2}(1-p)^{2}G''(p)+p(1-p)(1-2Rp)G'(p)+(1-R)p(1-Rp)G(p)\right]
\end{align*}
and hence the second order terms in \eqref{eq:hjb2} can be rewritten as:
\begin{eqnarray}
\lefteqn{p^{2}(1-p)^{2}G''(p)-2Rp^{2}(1-p)G'(p)-R(1-R)p^{2}G(p) \qquad} \nonumber \\
\hspace{20mm}&=&\sgn(1-p)|1-p|^{1-R}w(h)(w'(h)-1),  \nonumber \\
\lefteqn{-R(1-R)pG(p)+Rp(2p-1)G'(p)-p^{2}(1-p)G''(p) \qquad} \nonumber \\
&=&-\frac{|1-p|^{1-R}}{1-p}\sgn(1-p)\left(w'(h)w(h)-(1-R)w(h)\right), \nonumber \\
\lefteqn{p^{2}G''(p)+2RpG'(p)-R(1-R)G(p)} \nonumber \\
&=&\sgn(1-p)|1-p|^{-(1+R)}\left[w(h)w'(h)+(2R-1)w(h)-R(1-R)h\right]. \label{eq:concavityw}
\end{eqnarray}

Substituting back into \eqref{eq:hjb2}, and dividing through by $\sgn(1-p)|1-p|^{1-R}$ we obtain
\begin{align}
0&=\frac{b_{1}}{b_{4}}h|h|^{-1/R}\left|1-\frac{w(h)}{(1-R)h}\right|^{1-1/R}-\delta h \nonumber \\
&\qquad+r\left[(1-R)h-w(h)\right]+\alpha w(h) +\frac{\eta^{2}}{2}w(h)(w'(h)-1) \nonumber \\
&\qquad-\frac{\left\{\beta(1-R)h\left(1-\frac{w(h)}{(1-R)h}\right) -\eta\rho \left[w'(h)w(h)-(1-R)w(h)\right]\right\}^{2}}{2\left[w(h)w'(h)+(2R-1)w(h)-R(1-R)h\right]}.
\label{eq:hjb3}
\end{align}

Recall the definitions $W(h)=\frac{w(h)}{(1-R)h}$, $N=W^{-1}$ and $n(q)=|N(q)|^{-1/R}|1-q|^{1-1/R}$. Then $w(N(q))=(1-R)N(q)W(N(q))=(1-R)qN(q)$. Put $h=N(q)$ in \eqref{eq:hjb3} and divide by $h$. Then we have
\begin{align}
0&=\frac{b_{1}}{b_{4}}n(q)- \delta + r(1-R)(1-q)+\alpha(1-R)q +\frac{\eta^{2}}{2}(1-R)\left[qw'(N(q))-q\right] \nonumber \\
&\qquad -\frac{1-R}{2}\frac{\left\{\beta(1-q) -\eta\rho \left[qw'(N(q))-(1-R)q\right]\right\}^{2}}{qw'(N(q))+(2R-1)q-R}.
\label{eq:hjb4}
\end{align}

Recall the definitions of the auxiliary constants $(b_{i})_{i=1,2,3,4}$ given at the very start of Section~\ref{ssec:reduce}. Rearranging \eqref{eq:hjb4} and multiplying by $b_4$ 
\begin{eqnarray}
0&=&(1-R)q^{2}(w'(N(q)))^{2} \nonumber \\
&&\qquad+\left[b_{1}n(q)-\left[b_{1}+b_{2}R(1-R)\right]+(b_{3}+2R-2)(1-R)q\right]qw'(N(q)) \nonumber \\
&&\qquad+\left[(2R-1)(b_{3}-1)+R^{2}(1-b_{2})\right](1-R)q^{2} \nonumber \\
&&\qquad+\left[(1-2R)b_{1}+R(1-R)b_{2}-R(1-R)b_{3}\right]q \nonumber  \\
&&\qquad+b_{1}R+b_{1}\left[(2R-1)q-R \right]n(q) \nonumber \\
&=:& A(qw'(N(q)))^{2}+B(qw'(N(q)))+C. \label{eq:ABC}
\end{eqnarray}
This can be viewed as a quadratic equation in $qw'(N(q))$. Note that the coefficients $A$, $B$, $C$ depend on the market parameters only through the auxiliary parameters $b_1$, $b_2$, $b_3$.

We want the root corresponding to $V_{xx}<0$. This is equivalent to
\begin{equation}
 \frac{1}{1-R} p^2 G''(p) + \frac{2R}{1-R} pG'(p) -RG(p) < 0 .
\label{eq:VxxG}
\end{equation}
Using \eqref{eq:concavityw} and the fact that $\sgn(1-p) = \sgn(h)$, and multiplying \eqref{eq:VxxG} by $|1-p|^{R+1}/|h|$ we find we want the solution for which
\begin{equation}
  \frac{1}{(1-R)} \frac{1}{h} \{ w(h) w'(h) + (2R-1) w(h) - R(1-R) h \} =  \{ q w'(N(q)) + (2R-1) q - R \} < 0 .\label{eq:wconc}
\end{equation}

Consider \eqref{eq:hjb4}  and write $u = q w'(N(q))$. Then for fixed $q$ and $n(q)$,  \eqref{eq:hjb4} is of the form
$(1-R)a_1 u- a_2 = (1-R)\frac{(a_3 u+a_4)^2}{(u-a_5)}$ where $(a_i)_{1 \leq i \leq 5}$ are constants with $a_1 > a_3^2$ and $a_5=R-(2R-1)q$. It is easily seen that this equation has two solutions, one on each side of
$u=a_5$, and that from \eqref{eq:wconc} the one we want is the smaller root.
Thus
\begin{align*}
q w'(N(q))=\frac{-B-\sgn(A)\sqrt{B^{2}-4AC}}{2A}.
\end{align*}
where $A$, $B$, $C$ are the constants in \eqref{eq:ABC}. Note that $\sgn(A) = \sgn (1-R)$.
Then, we have
\begin{align*}
\frac{n'(q)}{n(q)}&=\frac{1-R}{R(1-q)}-\frac{1}{R}\frac{N'(q)}{N(q)} \\
&=\frac{1-R}{R(1-q)}-\frac{1-R}{R}\frac{q}{qw'(N(q))-(1-R)q^2} \\
&=\frac{1-R}{R(1-q)}-\frac{1-R}{R}\frac{2Aq}{-B-\sgn(A)\sqrt{B^{2}-4AC}-2A(1-R)q^2}.
\end{align*}
After some algebra, we arrive at
\begin{align*}
n'(q)=\frac{(1-R)n(q)}{R(1-q)}-\frac{2(1-R)^{2}qn(q)/R}{2(1-R)(1-q)\left[(1-R)q+R\right]-\varphi(q,n(q))-\sgn(1-R)\sqrt{\varphi(q,n(q))^{2}+E(q)^{2}}}.
\end{align*}

\section{Continuity and smoothness of the candidate value function}
\label{app:prop}

\begin{proof}[Proof of Case (i) of Proposition \ref{prop:def_g}]
We have
\begin{eqnarray*}
\int_{N(q_{*})}^{N(q^{*})}\frac{du}{w(u)}-\int_{p_{*}}^{p^{*}}\frac{du}{u(1-u)}
& = & \int_{q_{*}}^{q^{*}}\left(\frac{N'(u)}{(1-R)uN(u)}-\frac{1}{u(1-u)}\right)du+\int_{q_{*}}^{q^{*}}\frac{du}{u(1-u)}-\int_{p_{*}}^{p^{*}}\frac{du}{u(1-u)} \\
&=&\int_{q_{*}}^{q^{*}}\left(-\frac{R}{u(1-R)}\frac{O(u,n(u))}{n(u)}\right)du -\ln(1+\xi) \\
&=&0
\end{eqnarray*}
using \eqref{eq:trancost} and this establishes the equivalence of \eqref{eq:def_h_1} and \eqref{eq:def_h_2}.

Suppose we have a solution $(n(\cdot),q_{*},q^{*})$ to \eqref{eq:freebound} with $n$ being strictly positive. Let $N(q)=\sgn(1-q)n(q)^{-R}|1-q|^{R-1}$, $W=N^{-1}$ and $w(h)=(1-R)hW(h)$. We set $G^C(p) = \sgn(1-p) |1-p|^{1-R} h(p)$ where $h$ solves $\frac{dh}{dp} = \frac{w(h)}{p(1-p)}$.
For notational convenience (and to allow us to write derivatives as superscripts) write $G$ as shorthand for $G^C$.

First we check that $G$ is $C^2$. Outside the no-transaction interval this is immediate from the definition, and on $(p_*,p^*)$ it follows from the fact that $n$ and $n'$ are continuous. This property is inherited by the pair $(w,w')$ and then on integration by the trio $(h, h', h'')$ and finally $(G, G', G'')$.

It remains to check the continuity of $G$, $G'$ and $G''$ at $p_*$ and $p^*$. We prove the continuity at $p_*$; the proofs at $p^*$ are similar.
Using $\frac{1-q^*}{1-p^*} = \frac{1}{1+\lambda p^*}$ for the penultimate equivalence, we have
\begin{eqnarray*}
G(p_{*}+)&=&\sgn(1-p_{*})|1-p_{*}|^{1-R}h(p_{*}) \\
&=&\sgn(1-p_{*})|1-p_{*}|^{1-R}\sgn(1-q_{*})n(q_{*})^{-R}|1-q_{*}|^{R-1} \\
&=&n(q_{*})^{-R}(1+\lambda p_{*})^{1-R} = G(p_{*}-) .
\end{eqnarray*}
Then continuity of $G'$ at $p_*$ follows from \eqref{eq:Ghw} where
\[
G(p_*+)-\frac{p_*G'(p_* +)}{1-R} =  |1-p_*|^{-R} h_* (1-W(h_{*}))=
\frac{G(p_*+)}{1-p_*} (1 - q_*) = \frac{G(p_*)}{1 + \lambda p_*} =  G(p_*-)-\frac{p_*G'(p_* -)}{1-R}. \]
Finally, from \eqref{eq:concavityw},
\begin{eqnarray*}
\lefteqn{p_*^{2}G''(p_*+)+2Rp_*G'(p_*+)-R(1-R)G(p_*+)} \\
&=&\frac{G(p_*+)}{(1-p_*)^{2} h_{*}} \left[w(h_*)w'(h_*)+(2R-1)w(h_*)-R(1-R)h_*\right] \\
&=& - R(1-R)G(p_*) \left( \frac{1-q_*}{1-p_*} \right)^{2} = -R(1-R)\frac{G(p_*)}{(1 + \lambda p_*)^2} \\
& = & p_*^{2}G''(p_*-)+2Rp_*G'(p_*-)-R(1-R)G(p_*-)
\end{eqnarray*}
and we conclude that $G''$ is continuous at $p=p_{*}$.

Now we argue that $\frac{(x+y \theta)^{1-R}}{1-R}G( \frac{y \theta}{x+y \theta})$ is strictly increasing and strictly concave. Outside $[p_*,p^*]$ this is immediate form the definition. On $[p_*,p^*]$ the increasing property will follow if $G(p) - \frac{pG'(p)}{1-R}>0$. But this is trivial since
\begin{align*}
G(p) - \frac{pG'(p)}{1-R}=|1-p|^{-R}h(1-W(h))=|1-p|^{-R}N(q)(1-q)=|1-p|^{-R}|1-q|^{R}n(q)^{-R}>0.
\end{align*}
Meanwhile, $\frac{(x+y \theta)^{1-R}}{1-R}G( \frac{y \theta}{x+y \theta})$ is concave on $[p_*,p^*]$ is equivalent to \eqref{eq:VxxG}, or, by the analysis leading to \eqref{eq:wconc} to $qw'(N(q)) + (2R-1)q-R<0$. But this follows from our choice of root in \eqref{eq:ABC}.

\end{proof}

\begin{proof}[Proof of Case (ii) of Proposition \ref{prop:def_g}]
Note that the integrand of $\int_{q_{*}}^{q^{*}}\left(\frac{R}{q(1-R)}\frac{O(q,n(q))}{n(q)}\right)dq$ is everywhere negative and therefore $\int_{q_{*}}^{1}\left(-\frac{R}{q(1-R)}\frac{O(q,n(q))}{n(q)}\right)dq$ exists in $[0, \ln(1+\xi)]$. Hence $-\ln(1+\xi) \leqslant a \leqslant  \ln(1+\xi)$.

For $p\neq 1$, the $C^{2}$ smoothness of $G=G^C$ follows as in the first case of Proposition \ref{prop:def_g}. We will focus on the case of $p=1$.

Suppose first that $p_{*}<1<p^{*}$. Continuity of $G$ and $G'$ at $p=1$ can be established if we can show that both
\begin{align}
\lim_{p\to 1}\frac{1}{G(p)}\left(G(p)-\frac{pG'(p)}{1-R}\right)^{1-1/R}=n(1)
\label{eq:first_con}
\end{align}
and
\begin{align}
\lim_{p\to 1}\frac{pG'(p)}{(1-R)G(p)}=1-e^{a}.
\label{eq:second_con}
\end{align}
Substituting \eqref{eq:second_con} into \eqref{eq:first_con} we recover the given value of $G(1)$.

Using \eqref{eq:Vxdef} and the equivalence of $p\to 1$ and $q\to 1$ we have $\frac{1}{G(p)}\left(G(p)-\frac{pG'(p)}{1-R}\right)^{1-1/R} = |h|^{-1/R} |1-W(h)|^{1-1/R} = |N(q)|^{-1/R}|1-q|^{1-1/R}=n(q) \rightarrow n(1)$ and \eqref{eq:first_con} holds.

For \eqref{eq:second_con} we have,
\begin{align*}
\frac{1-W(h(p))}{1-p}&=\frac{(1-R)h(p)-p(1-p)h'(p)}{(1-R)(1-p) h(p)}=1-\frac{pG'(p)}{(1-R)G(p)}.
\end{align*}

Suppose $p<1$. Then using the definition of $h(p)$,
\begin{align*}
0&=\int_{N(q_{*})}^{h(p)}\frac{du}{w(u)}-\int_{p_{*}}^{p}\frac{du}{u(1-u)} \\
&=\int_{q_{*}}^{W(h(p))}\frac{N'(q)dq}{(1-R)qN(q)}-\int_{p_{*}}^{p}\frac{du}{u(1-u)} \\
&=\int_{q_{*}}^{W(h(p))}\left(\frac{N'(q)}{(1-R)qN(q)}-\frac{1}{q(1-q)}\right)dq+\int_{q_{*}}^{W(h(p))}\frac{dq}{q(1-q)}-\int_{p_{*}}^{p}\frac{du}{u(1-u)} \\
&=\int_{q_{*}}^{W(h(p))}\left(-\frac{R}{u(1-R)}\frac{O(u,n(u))}{n(u)}\right)du-\int_{p_{*}}^{q_{*}}\frac{du}{u(1-u)}-\int_{W(h(p))}^{p}\frac{dq}{q(1-q)} \\
&=\int_{q_{*}}^{W(h(p))}\left(-\frac{R}{u(1-R)}\frac{O(u,n(u))}{n(u)}\right)du-\ln(1+\lambda)-\ln\left(\frac{p}{W(h(p))}\frac{1-W(h(p))}{1-p}\right).
\end{align*}
Letting $p\uparrow 1$ and using the fact that $\lim_{p\to 1}W(h(p))=1$, we obtain
\begin{align}
\lim_{p\uparrow 1}\frac{1-W(h(p))}{1-p}=e^a.
\label{eq:limit_exp}
\end{align}
A similar calculation for $p>1$ gives $\lim_{p\downarrow 1}\frac{W(h(p))-1}{p-1}=e^a$ as well. Hence \eqref{eq:second_con} holds. As a byproduct, we can establish
\begin{align*}
\lim_{p\to 1}G'(p)=(1-R)(1-e^{a})G(1)=(1-R)(1-e^{a})n(1)^{-R}e^{-(1-R)a}.
\end{align*}

Consider now continuity of $G''$ at $p=1$. We show $\lim_{p\to 1}G''(p)$ exists. Consider:
\begin{align*}
\frac{[(1-R)G(p)-pG'(p)]^{2}}{G(p)[p^{2}G''(p)+2RpG'(p)-R(1-R)G(p)]}&=\frac{(1-R)^{2}h(1-W(h))^{2}}{w(h)w'(h)+(2R-1)w(h)-R(1-R)h} \\
&=\frac{(1-R)(1-q)^{2}}{(1-R)qN(q)/N'(q)-(1-q)[R+(1-R)q]} \\
&=\frac{(1-R)\left[1-R-R(1-q)n'(q)/n(q)\right]}{R[R+(1-R)q]n'(q)/n(q)-R(1-R)} .
\end{align*}
Then,
\begin{align}
\lim_{p\to 1}\frac{[(1-R)G(p)-pG'(p)]^{2}}{G(p)[p^{2}G''(p)+2RpG'(p)-R(1-R)G(p)]} &= \lim_{q\to 1}\frac{(1-R)\left[1-R-R(1-q)n'(q)/n(q)\right]}{R\{[R+(1-R)q]n'(q)/n(q)-(1-R)\}} \nonumber \\
&=\frac{(1-R)^{2}}{R[n'(1)/n(1)-(1-R)]}.
\label{eq:2nd_deriv_expression}
\end{align}
Note that $n'(1)/n(1)-(1-R)\neq 0$ since $\sgn(n'(1))=-\sgn(1-R)$. The limit is thus always well defined and can be used to obtain an expression for $\lim_{p\to 1}G''(p)$.

Since $G$ is $C^2$ and \eqref{eq:VxxG} holds for both $p<1$ and $p>1$ it follows that \eqref{eq:VxxG} holds at $p=1$ also and $\frac{(x+y \theta)^{1-R}}{1-R}G( \frac{y \theta}{x+y \theta})$ is concave on $[p_*,p^*]$.

Finally we consider the case where $p_{*}=1$ or $p^{*}=1$. Suppose we are in the former scenario. Then to show the continuity of $G$ at $p_{*}=1$ it is sufficient to show that
\begin{align*}
n(q_{*})^{-R}\left(1+\lambda \right)^{1-R}=n(1)^{-R}e^{-(1-R)a}.
\end{align*}
But $q_{*}=1$ when $p_{*}=1$ and thus $a=-\ln(1+\lambda)$. The above expression then holds immediately. Values of $G'(1)$ and $G''(1)$ can again be inferred from \eqref{eq:second_con} and \eqref{eq:2nd_deriv_expression}.  A similar result follows in the case $p^{*}=1$.

\end{proof}

\section{The candidate value function and the HJB equation}
\label{subsect:candidate_valfun}

In this section we verify that the candidate value function given in Proposition~\ref{prop:def_g} solves the HJB variational inequality
\begin{align}
\min\left(-\sup_{c>0,\pi}\mathcal{L}^{c,\pi}V^{C},-\mathcal{M}V^{C},-\mathcal{N}V^{C}\right)=0
\label{eq:hjb_ineq}
\end{align}
where $\mathcal{L}$, $\mathcal{M}$ and $\mathcal{N}$ are the operators
\begin{align*}
\mathcal{L}^{c,\pi}f&:=\frac{c^{1-R}}{1-R}-cf_{x}+\frac{\sigma^{2}}{2}f_{xx}\pi^{2}+((\mu-r) f_{x}+\sigma\eta\rho f_{xy}y)\pi \\
&\qquad+r f_{x} x+\alpha f_{y} y+\frac{\eta^{2}}{2}f_{yy}y^{2}-\delta f, \\
\mathcal{M}f&:=f_{\theta}-(1+\lambda)yf_{x},\\
\mathcal{N}f&:=(1-\gamma)yf_{x}-f_{\theta}.
\end{align*}

Note that for $f=f(x,y,\theta)$ which is strictly increasing and concave in $x$ we have
\begin{align*}
\mathcal{L}^{*}f:=\sup_{c>0,\pi}\mathcal{L}^{c,\pi}f=\frac{R}{1-R}f_{x}^{1-1/R}+rx f_{x}+\alpha yf_{y}+\frac{\eta^{2}}{2}y^{2}f_{yy}-\frac{(\beta f_{x}+\eta\rho yf_{xy})^{2}}{2f_{xx}}-\delta f
\end{align*}
and thus it is equivalent to show that $\min\left(-\mathcal{L}^{*}V^{C},-\mathcal{M}V^{C},-\mathcal{N}V^{C}\right)=0$. From construction of $V^{C}$, it is trivial that $\mathcal{L}^{*}V^{C}=0$, $\mathcal{M}V^{C}=0$ and $\mathcal{N}V^{C}=0$ on the no-transaction region, purchase-region and sale-region respectively. Hence it remains to show that
\begin{align*}
\begin{cases}
\mathcal{L}^{*}V^{C}\leqslant 0,\quad \mathcal{N}V^{C}\leqslant 0,&-1/\lambda\leqslant p<p_{*}; \\
\mathcal{M}V^{C}\leqslant 0,\quad \mathcal{N}V^{C}\leqslant 0 ,& p_{*}\leqslant p\leqslant p^{*}; \\
\mathcal{L}^{*}V^{C}\leqslant 0,\quad \mathcal{M}V^{C}\leqslant 0,&p^{*}<p\leqslant1/\gamma.
\end{cases}
\end{align*}

On the purchase region $p\in[-1/\lambda,p_{*})$, direct substitution reveals that
\begin{align*}
\mathcal{N}V^{C}=-\left(\frac{b_{1}}{Rb_{4}}\right)^{-R}n(q_{*})^{-R}(\lambda+\gamma)y(x+y\theta)^{-R}(1+\lambda p)^{-R}\leqslant 0,
\end{align*}
and
\begin{align*}
\mathcal{L}^{*}V^{C}=\frac{R(x+y\theta)^{1-R}}{1-R}\left(\frac{b_{1}}{Rb_{4}}\right)^{1-R}(1+\lambda p)^{1-R}n(q_{*})^{-R}\left(m(q_{*})-m\left(\frac{(1+\lambda)p}{1+\lambda p}\right)\right)\leqslant 0
\end{align*}
where we have used the facts that $n(q_{*})=m(q_{*})$, $\frac{(1+\lambda)p}{1+\lambda p}<\frac{(1+\lambda)p_{*}}{1+\lambda p_{*}}=q_{*}$ and the quadratic $m(q)$ is decreasing (respectively increasing) over $q<q_{*}<q_{M}$ when $R<1$ (respectively $R>1$). Similar calculations can be performed on the sale region $p\in(p^{*},1/\gamma]$ to show that $\mathcal{M}V^{C}\leqslant 0$ and $\mathcal{L}^{*}V^{C}\leqslant 0$.




Now we show that $\mathcal{M}V^{C}\leqslant 0$  on the no-transaction region $p\in[p_{*},p^{*}]$. The inequality $\mathcal{N}V^{C}\leqslant 0$ can be proved in an identical fashion. Again writing $G$ as shorthand for $G^C$, we have
\begin{align*}
\mathcal{M}V^{C}&=V^{C}_{\theta}-(1+\lambda)yV^{C}_{x} =\frac{pV^{C}}{\theta}\left[(1+\lambda p)\frac{G'(p)}{G(p)}-\lambda(1-R)\right].
\end{align*}
Since $\sgn(V^{C})=\sgn(1-R)$, it is necessary and sufficient to show
\begin{align*}
\sgn(1-R)\left[(1+\lambda p)\frac{G'(p)}{G(p)}-\lambda(1-R)\right]\leqslant 0.
\end{align*}
But $G(p)=\sgn(1-p)h(p)|1-p|^{1-R}$ for $p\neq 1$,  and then
\begin{align*}
\frac{G'(p)}{G(p)}=\frac{h'(p)}{h(p)}-\frac{1-R}{1-p}=\frac{w(h)}{h(p)p(1-p)}-\frac{1-R}{1-p}=\frac{1-R}{1-p}\left(\frac{W(h)}{p}-1\right)
\end{align*}
and the required inequality becomes
\begin{align}
\frac{1-W(h)}{1-p}\geqslant \frac{1}{1+\lambda p}.
\label{eq:MV_ineq}
\end{align}
We are going to prove \eqref{eq:MV_ineq} for $p\in[p_{*},p^{*}]\setminus\{1\}$. Then $\mathcal{M}V^{C}\leqslant 0$ will hold at $p=1$ as well by smoothness of $V^{C}$.

By construction $q = W(h(p))$. Since $W$ is monotonic and $h$ is monotonic except possibly at $p=1$ it follows that $q$ is an increasing function of $p$. Then, starting from the identity
\[ \int_{N(q_*)}^{N(q)} \frac{dh}{w(h)} = \int_{p_*}^{p} \frac{du}{u(1-u)} \]
and following the substitutions leading to \eqref{eq:tci}, we find
\begin{align*}
\int_{q_{*}}^{q}\left(-\frac{R}{u(1-R)}\frac{O(u,n(u))}{n(u)}\right)du= - \int_{q_*}^{q} \frac{dv}{v(1-v)} + \int_{p_*}^{p} \frac{du}{u(1-u)}.
\end{align*}
Since the expression on the left hand side is increasing in $q$, we deduce
\begin{align*}
\frac{1}{q(1-q)} \frac{dq}{dp}\leqslant \frac{1}{p(1-p)}.
\end{align*}

Define $\chi(p):=\frac{(1+\lambda)p}{1+\lambda p}$. then $\chi$ is a solution to the ODE $\chi'(p)=\varrho(p,\chi(p))$ where $\varrho(p,y)=\frac{y(1-y)}{p(1-p)}$.
Note that $\chi(p_{*})=\frac{(1+\lambda)p_{*}}{1+\lambda p_{*}}=q_{*}=q(p_{*})$.

Suppose $p_{*}<p^{*}< 1$. Then for $p<1$ and in turn $q=q(p)=W(h(p))<1$ we have
$q'(p)\leqslant \varrho(p,q(p))$, and we conclude $q(p)\leqslant \chi(p)$ for $p_{*}\leqslant p<p^* \leqslant 1$. Then
\begin{align*}
1-W(h(p))=1-q(p)\geqslant1-\chi(p)= \frac{1-p}{1+\lambda p}
\end{align*}
which establishes \eqref{eq:MV_ineq}. If instead $1<p_{*}<p^{*}$, we can arrive at the same result by showing $q(p)\geqslant \chi(p)$ for $1<p_{*}\leqslant p$ and in turn $\frac{dq}{dp}\geqslant \frac{q(q-1)}{p(p-1)}$.

It remains to consider the case of $p_{*}\leqslant 1\leqslant p^{*}$. The only issue is that the comparison of derivatives of $q(p)$ and $\chi(p)$ may not be trivial at $p=1$ because of the singularity in $\varrho(p,y)$. But by direct computation, we find $\chi'(1)=\frac{1}{1+\lambda}$. On the other hand,
\begin{align*}
q'(1-)=\lim_{p\uparrow 1}\frac{1-q(p)}{1-p}=\lim_{p\uparrow 1}\frac{1-W(h(p))}{1-p}=e^{a}
\end{align*}
due to \eqref{eq:limit_exp} and similarly we have $q'(1+)=e^{a}$. Then $q'(1)$ is well-defined, and moreover since $a>-\ln(1+\lambda)$ we have
\begin{align*}
q'(1)=e^{a}>1/(1+\lambda)=\chi'(1).
\end{align*}
Together with the fact that $q(1)=1=\chi(1)$, we must have that $q(p)$ is an upcrossing of $\chi(p)$ at $p=1$. From this we conclude $q(p)\leqslant \chi(p)$ on $p\in[p_{*},1)$ and $\chi(p)\leqslant q(p)$ on $p\in(1,p^{*}]$.  \eqref{eq:MV_ineq} then follows.

\section{Proof of the main results}
\label{app:pfmain}
\begin{proof}[Proof of Theorems \ref{thm:wellposed} and \ref{thm:valfun}]

We prove the two theorems together. Suppose we are in the well-posed cases. From the analysis in Section \ref{sect:fbp}, there exists a solution $(n(\cdot),q_{*},q^{*})$ to the free boundary value problem with $n$ being strictly positive. By the $C^{2}$ smoothness of $G^{C}$, $V^{C}$ is $C^{2\times 2 \times 1}$. Moreover, in Appendices \ref{app:prop} and \ref{subsect:candidate_valfun} we saw that $V^{C}$ is a strictly concave function in $x$ solving the HJB variational inequality \eqref{eq:hjb_ineq}.

Let $M_{t}:=\int_{0}^{t}e^{-\delta s}\frac{C_{s}^{1-R}}{1-R}ds+e^{-\delta t}V^{C}(X_{t},Y_{t},\Theta_{t})$. Applying Ito's lemma, we obtain
\begin{align*}
M_{t}&=M_{0}+\int_{0}^{t}e^{-\delta s}\mathcal{L}^{C_{s},\Pi_{s}}V^{C}ds+\int_{0}^{t}e^{-\delta s}\mathcal{M}V^{C}d\Phi_{s}+\int_{0}^{t}e^{-\delta s}\mathcal{N}V^{C}d\Psi_{s} \\
&\qquad+\int_{0}^{t}e^{-\delta s}\sigma V^{C}_{x}\Pi_{s}dB_{s}+\int_{0}^{t}e^{-\delta s}\eta V^{C}_{y} Y_{s} dW_{s}\\
&\leqslant M_{0}+\int_{0}^{t}e^{-\delta s}\sigma V^{C}_{x}\Pi_{s}dB_{s}+\int_{0}^{t}e^{-\delta s}\eta V^{C}_{y} Y_{s} dW_{s}.
\end{align*}

Suppose $R<1$. Then $M_{t}\geqslant 0$, and the sum of the stochastic integrals is a local martingale bounded below by $-M_{0}$ and in turn it is a supermartingale. Thus $\mathbb{E}(M_{t})\leqslant M_{0}=V^{C}(x,y,\theta)$ which gives
\begin{align*}
\mathbb{E}\left(\int_{0}^{t}e^{-\delta s}\frac{C_{s}^{1-R}}{1-R}ds\right)\leqslant V^{C}(x,y,\theta)-\mathbb{E}\left(e^{-\delta t}V^{C}(X_{t},Y_{t},\Theta_{t})\right)\leqslant V^{C}(x,y,\theta).
\end{align*}
On sending $t\to\infty$, we obtain $\mathbb{E}\left(\int_{0}^{\infty}e^{-\delta s}\frac{C_{s}^{1-R}}{1-R}ds\right)\leqslant V^{C}$ by monotone convergence and thus $V\leqslant V^{C}$ since $C$ is arbitrary.

If $R>1$, then the above argument does not go through directly since the local martingale will not be bounded below. But using the argument of \cite{DavisNorman:90}, we can consider a perturbed candidate value function which is bounded on the no-transaction region and define a version of the value process $M$ which will be a supermartingale. The result can be obtained by considering the limit of the perturbed candidate value function.

To show $V^{C}\leqslant V$, it is sufficient to demonstrate the existence of an investment/consumption strategy which attains the value $V^{C}$. Suppose the initial value $(x,y\theta)$ is such that $\frac{y\theta}{x+y \theta} = p\in[p_{*},p^{*}]$. Define feedback controls $C^{*}=(C^{*}_{t})_{t\geqslant 0}$ and  $\Pi^{*}=(\Pi^{*}_{t})_{t\geqslant 0}$ with $C^{*}_{t}=C^{*}(X_{t},Y_{t},\Theta_{t})$ and $\Pi^{*}_{t}=\Pi^{*}(X_{t},Y_{t},\Theta_{t})$ where
\begin{align*}
C^{*}(x,y,\theta):=[V_{x}^{C}(x,y,\theta)]^{-\frac{1}{R}},\quad \Pi^{*}(x,y,\theta):=-\frac{(\mu-r)V^{C}_{x}(x,y,\theta_{t})+\sigma\eta\rho y V^{C}_{xy}(x,y,\theta_{t})}{\sigma^{2}V^{C}_{xx}(x,y,\theta_{t})},
\end{align*}
and $\Theta^{*}=(\Theta^{*}_{t})_{t\geqslant 0}$ a finite variation, local time strategy in form of $\Theta^{*}_{t}=\theta+\Phi_{t}^{*}-\Psi_{t}^{*}$ which keeps $P_{t}$ within $(p_{*},p^{*})$. Let $X^{*}$ be the liquid wealth process evolving under these controls. Now since $(X^{*},Y\Theta^{*})$ is always located in the no-transaction wedge, this strategy is clearly admissible.

Let $M^{*}$ be the process $M^{*}=(M^{*}_{t})_{t\geqslant 0}$ evolving under this controlled system. Then
\begin{align*}
M^{*}_{t}&=M^{*}_{0}+\int_{0}^{t}e^{-\delta s}\mathcal{L}^{C^{*}_{s},\Pi^{*}_{s}}V^{C}ds+\int_{0}^{t}e^{-\delta s}\mathcal{M}V^{C}d\Phi^{*}_{s}+\int_{0}^{t}e^{-\delta s}\mathcal{N}V^{C}d\Psi^{*}_{s} \\
&\qquad+\int_{0}^{t}e^{-\delta s}\sigma V^{C}_{x}\Pi^{*}_{s}dB_{s}+\int_{0}^{t}e^{-\delta s}\eta V^{C}_{y} Y_{s} dW_{s}\\
&=:M^{*}_{0}+N_{t}^{1}+N_{t}^{2}+N_{t}^{3}+N_{t}^{4}+N_{t}^{5}.
\end{align*}
By construction of $C^{*}$ and $\Pi^{*}$, $N_{t}^{1}=0$. Moreover, $\Phi^{*}$ is carried by the set $\{P_{t}=p_{*}\}$ over which $\mathcal{M}V^{C}_{s}=0$. Hence $N_{t}^{2}=0$, and similarly $N_{t}^{3}=0$. Following ideas similar to Davis and Norman~\cite{DavisNorman:90}, it can be shown (see Tse~\cite{tse:16}) that the local-martingale stochastic integrals $N^{4}$ and $N^{5}$ are martingales. Then on taking expectation we have
\begin{align}
\mathbb{E}\left(\int_{0}^{t}e^{-\delta s}\frac{(C^{*}_{s})^{1-R}}{1-R}ds\right)+\mathbb{E}(e^{-\delta t}V^{C}(X^{*}_{t},Y_{t},\Theta^{*}_{t}))=\mathbb{E}(M^{*}_{t})=M^{*}_{0}=V^{C}.
\label{eq:opteq}
\end{align}
Further, it can also be shown (see Tse~\cite{tse:16}) that $\lim_{t\to\infty}\mathbb{E}(e^{-\delta t}V^{C}(X^{*}_{t},Y_{t},\Theta^{*}_{t}))=0$. Then letting $t\to\infty$ in \eqref{eq:opteq} gives
\begin{align*}
V^{C}&=\mathbb{E}\left(\int_{0}^{\infty}e^{-\delta s}\frac{(C^{*}_{s})^{1-R}}{1-R}ds\right)\leqslant \sup_{(C,\Pi,\Theta)\in\mathcal{A}(0,x,y,\theta)}\mathbb{E}\left(\int_{0}^{\infty}e^{-\delta s}\frac{C_{s}^{1-R}}{1-R}ds\right)=V.
\end{align*}

Now suppose the initial value $(x,y\theta)$ is such that $p<p_{*}$. Then consider a strategy of purchasing $\phi=\frac{xp_{*}-(1-p_{*})y\theta}{y(1+\lambda p_{*})}$ number of shares at time zero such that the post-transaction proportional holding in the illiquid asset is $\frac{y(\theta+\phi)}{x+y(\theta+\phi)-y(1+\lambda)\phi}=p_{*}$, and then follow the investment/consumption strategy $(C^{*},\Pi^{*},\Theta^{*})$ as in the case of $p\in[p_{*},p^{*}]$ thereafter. By construction of $V^{C}$, $V^{C}(x,y,\theta)=V^{C}(x-y(1+\lambda)\phi,y,\theta+\phi)$. Using \eqref{eq:opteq} we have
\begin{align*}
\mathbb{E}\left(\int_{0}^{t}e^{-\delta s}\frac{(C^{*}_{s})^{1-R}}{1-R}ds\right)+\mathbb{E}(e^{-\delta t}V^{C}(X^{*}_{t},Y_{t},\Theta^{*}_{t}))=V^{C}(x-y(1+\lambda)\phi,y,\theta+\phi)=V^{C}(x,y,\theta)
\end{align*}
and from this we can conclude $V^{C}\leqslant V$. Similar argument applies for initial value $p>p^{*}$.

Now we consider the set of parameters which leads to unconditional ill-posedness. It is sufficient to show that the problem without the liquid asset (which is the classical transaction cost problem involving one single risky asset only) is ill-posed. Note that $\ell(1)\leqslant0$ is equivalent to $b_{3}\geqslant \frac{b_{1}}{1-R}+b_{2}R$ and this inequality can be restated as $\alpha\geqslant \frac{1}{2}\eta^{2}R+\frac{\delta}{1-R}$. But this is exactly the ill-posedness condition in the one risky asset case. See \cite{HobsonTseZhu:16} or \cite{ChoiSirbuZitkovic:13}.

Finally we consider the conditionally well-posed case. From the discussion in Section \ref{sect:fbp}, it is clear that as long as $\xi>\overline{\xi}$ there still exists $(n(\cdot),q_{*},q^{*})$ a solution to the free boundary value problem and thus one could show $V^{C}=V$ following the same argument in the proof for the unconditionally well-posed cases. Moreover, from Lemma \ref{lemma:onto} we can see that $n(\cdot)\downarrow 0$ as $\xi\downarrow \overline{\xi}$, in turn $V^{C}\to\infty$ from its construction. But $V\geqslant V^{C}$ and thus we conclude $V\to\infty$ as $\xi\downarrow \overline{\xi}$. This shows the ill-posedness of the problem at $\xi=\overline{\xi}$, and using the monotonicity of $V$ in $\xi$ this conclusion extends to any $\xi\leqslant \overline{\xi}$.

\end{proof}

\section{The first order differential equation}
\label{app:ode}
For convenience, we recall some notations, and introduce some more:
\begin{eqnarray}
m(q) & = & \frac{R(1-R)}{b_{1}}q^{2}-\frac{b_{3}(1-R)}{b_{1}}q+1, \nonumber \\
\ell(q) & = & m(q)+\frac{1-R}{b_{1}}q(1-q)+\frac{(b_{2}-1)R(1-R)}{b_{1}}\frac{q}{(1-R)q+R}, \nonumber \\
\varphi(q,n) & = & b_{1}(n-1)+(1-R)(b_{3}-2R)q+(2-b_{2})R(1-R), \nonumber \\
E(q)^{2} & = & 4R^{2}(1-R)^{2}(b_{2}-1)(1-q)^{2}, \nonumber \\
v(q,n) & = & \varphi(q,n)-\sgn(1-R)\sqrt{\varphi(q,n)^{2}+E(q)^{2}} , \nonumber \\
D(q,n) & = & 2b_{1}[(1-R)q+R][n-m(q)]-q\left[v(q,n)-v(q,m(q))\right], \nonumber \\
A(q,n) & = & (\ell(q)-n)\left(2b_{1}[(1-R)q+R]-b_{1}q\left(1- \sgn(1-R)\frac{\varphi}{\sqrt{\varphi^{2}+E^{2}}}\right)\right)+D(q,n). \label{eq:Adef}
\end{eqnarray}

We begin with a useful lemma.
\begin{lemma}
$O(q,n)$ has an alternative expression
\begin{align}
O(q,n)=-\frac{(1-R)nD(q,n)}{2R(1-q)[(1-R)q+R]b_{1}[\ell(q)-n]} .
\label{eq:formOalt}
\end{align}
\label{lemma:different_expression_o}
\end{lemma}
\begin{proof}
Consider
\begin{eqnarray*}
\lefteqn{b_1(\ell(q) - n) + \varphi(q, n)} \\
 & = & R(1-R)q^2 - b_3(1-R)q + b_1 - b_1 n + (1-R)q(1-q) +
\frac{(b_2 - 1)R(1-R)q}{(1-R)q + R}\\
& & \hspace{5mm} + b_1 n - b_1 + b_3(1-R)q + R(1-R)[-2q + 2 - b_2]\\
& = & R(1-R)\left[(1 - q)^2 - (b_2 - 1) + \frac{(b_2 - 1)q}{(1 - R)q + R}\right] + (1-R)q(1-q) \\
& = & (1-R)(1-q)[R(1-q) + q] - \frac{(b_2 - 1)R^2(1-R)}{(1-R)q + R} (1-q).
\end{eqnarray*}
Then, noting that $(1-R)q + R= R(1-q) + q$,
\begin{eqnarray*}
\lefteqn{b_1 [(1-R)q + R](\ell(q) - n)} \\
& = & (1-R)(1-q)[R(1-q) + q]^2 - R^2(1-R)(b_2 -  1)(1-q) - \varphi(q,n)[R(1-q) + q],
\end{eqnarray*}
and multiplying by $4(1-R)(1-q)$,
\begin{eqnarray*}
\lefteqn{4b_1(1-R)(1-q)[(1-R)q + R](\ell(q) - n)} \\
& = & 4(1-R)^2(1-q)^2[R(1-q) + q]^2 - 4\varphi(q,n)(1-R)(1-q)[R(1-q) + q] + \varphi(q,n)^2 \\
& & \hspace{5mm} - \{\sgn(1-R)\}^2 \left(\varphi(q,n)^2 + 4R^2(1-R)^2(b_2 - 1)(1-q)^2\right) \\
& = & \left\{2(1-R)(1-q)[R(1-q)+q] - \varphi(q,n)\right\}^2 -\left\{\sgn(1-R)\right\}^2
\left\{\varphi(q,n)^2 + E(q)^2\right\}.
\end{eqnarray*}
Writing this last expression as the difference of two squares we find
\begin{align*}
&2(1-R)(1-q)[(1-R)q + R]- \varphi(q,n) -\sgn(1-R)\sqrt{\varphi(q,n)^2 + E(q)^2} \\
&\qquad =\frac{4b_1(1-R) (1-q) [(1-R)q + R](\ell(q) - n)}{2(1-R)(1-q)[R(1-q) + q] - v(q,n)}.
\end{align*}
Then
\begin{align*}
O(q,n)&=\frac{(1-R)n}{R(1-q)}-\frac{2(1-R)^{2}qn/R}{2(1-R)(1-q)\left[(1-R)q+R\right]-\varphi(q,n)-\sgn(1-R)\sqrt{\varphi(q,n)^{2}+E(q)^{2}}} \\
&=\lefteqn{\frac{(1-R)n}{R(1-q)} \left\{1 - \frac{(1-R)q(1-q)}{b_1 (\ell(q) - n)} + \frac{qv(q,n)}
{2b_1 [(1-R)q + R] (\ell(q) - n)}\right\} } \\
&=  \frac{(1-R)n \left\{2b_1 (\ell(q) - n) [(1-R)q + R]
-2[(1-R)q + R](1-R)q(1-q) + qv(q,n) \right\}}{2b_1R[(1-R)q + R](1-q) (\ell(q) - n)} \\
& = \frac{(1-R)n \left\{2b_1 [(1-R)q + R] \left[(\ell(q) - m(q)) - (n - m(q)) - \frac{(1-R)q(1-q)}{b_1}\right] + qv(q,n)\right\}}
{2b_{1}R(1-q)[(1-R)q + R](l(q) - n)}.
\end{align*}
The result then follows since
\[
2b_1 [(1-R)q + R] \left\{\ell(q) - m(q) - \frac{(1-R)q(1-q)}{b_1}\right\} = 2R(1-R)(b_2 - 1)q = -qv(q,m).
\]

\end{proof}

\begin{proof}[Proof of Lemma \ref{lemma:list}]


(1) Observe that
\begin{align*}
\ell(q)-m(q)&=\frac{1-R}{b_{1}}q(1-q)+\frac{(b_{2}-1)R(1-R)}{b_{1}}\frac{q}{(1-R)q+R} \\
&=\frac{(1-R)q}{b_{1}[(1-R)q+R]}P(q)
\end{align*}
where $P(q)=Rb_2 + (1-2R)q-(1-R)q^{2}$.
Hence the crossing points of $\ell(q)$ and $m(q)$ away from $q=0$ are given by the roots of $P(q)=0$ if such roots exist. Note that $P(-\frac{R}{1-R})=P(1)=R(b_{2}-1)>0$, since by assumption, $b_{2}>1$.

If $R<1$, then since $P$ is inverse U-shaped and $P(1)>0$ there must be two distinct solutions of the quadratic equation $P(q)=0$. As $0<P(1)=P(-R/(1-R))$, we must have $P(q)>0$ on $q\in[-R/(1-R),1]$, and the two roots must be found outside this interval. If $R>1$, the minima of $P(q)$ is given by $q_{P}:=\frac{2R-1}{2(R-1)}$. Note that $1<q_{P}<R/(R-1)$, and since $0<P(1)=P(R/(R-1))$, the root(s) of $P(q)=0$ must be contained on the interval $(1,R/(R-1))$ if they exist. The desired results can be established easily using these properties of $P$.

(2) The behaviour at $q=-R/(1-R)$ is only relevant for $R>1$ so we write this as $q=R/(R-1)$.
Note that $\ell$ explodes at $q =\frac{R}{R-1}$. It is sufficient to check the denominator of $O(q,n)$ is not equal to zero at $q=R/(R-1)$. Direct calculation gives
\[ [(1-R)q+R][\ell(q)-n] |_{q = \frac{R}{R-1}}
=-\frac{(b_{2}-1)R^{2}}{b_{1}} \]
and hence
\begin{equation}
\left. 2R(1-q)[(1-R)q+R]b_{1}[\ell(q)-n]\right|_{q= \frac{R}{R-1}}=\frac{2R^{3}(b_{2}-1)}{(R-1)}\neq0.
\label{eq:denO}
\end{equation}

(3) The following lemma records some useful identities.
\begin{lemma}
\begin{eqnarray*}
\varphi(q,m(q)) & = & R(1-R) \{ (1-q)^2 - (b_2-1) \}, \\
\varphi(q,\ell(q)) & = & (1-R)(1-q)\left\{(1-R)q+R-\frac{(b_{2}-1)R^{2}}{(1-R)q+R}\right\} ,\\
\varphi(1,n) & = & b_1(n - \ell(1)), \\
v(q,m(q)) & = & - 2R(1-R)(b_2-1), \\
v(q,\ell(q)) & = & \begin{cases}
-\frac{2R^{2}(1-R)(1-q)(b_{2}-1)}{(1-R)q+R},&(1-q)[(1-R)q+R]>0; \\
2(1-R)(1-q)[(1-R)q+R],&(1-q)[(1-R)q+R]<0,
\end{cases} \\
v(1,n) & = & 
\varphi(1,n)-\sgn(1-R)|\varphi(1,n)|.
\end{eqnarray*}
\end{lemma}

\begin{proof}
Most of these identities follow easily on substitution.
For $v(q,\ell(q))$ we have
\begin{align*}
v(q,\ell(q))   
&=(1-R)(1-q)\left\{(1-R)q+R-\frac{(b_{2}-1)R^{2}}{(1-R)q+R}\right\} \\
&\qquad -\sgn(1-R)\sqrt{(1-R)^{2}(1-q)^{2}\left\{(1-R)q+R+\frac{(b_{2}-1)R^{2}}{(1-R)q+R}\right\}^{2}} \\
&=(1-R)(1-q)\left\{(1-R)q+R-\frac{(b_{2}-1)R^{2}}{(1-R)q+R}\right\} \\
&\qquad -(1-R)|1-q|\left|(1-R)q+R+\frac{(b_{2}-1)R^{2}}{(1-R)q+R}\right| \\
\end{align*}
which simplifies to give the stated expression.
\end{proof}

Return to the proof of Part (3) of Lemma~\ref{lemma:list}.
Note that $\sgn(\varphi(1,n))=\sgn(n-\ell(1))$. Assume we are in the range $(1-R)n<(1-R)\ell(1)$. Then $\sgn(\varphi(1,n)) = - \sgn(1-R)$, $v(1,n)=2\varphi(1,n)$ and
\[ D(1,n) = 2b_{1}[n-m(1)]-v(1,n)+v(1,m(1)) =  2b_{1}[n-m(1)] - 2b_{1}[n-\ell(1)] +  2b_{1}[m(1)-\ell(1)] = 0. \]
Further, after some algebra we can show $\frac{\partial}{\partial q} D(q,n)|_{q=1} = - 2 b_1 R (n - m(1))$.

Consider $F(q,n)=\frac{O(q,n)}{n} = -\frac{(1-R)D(q,n)}{2R(1-q)[(1-R)q+R]b_{1}[\ell(q)-n]}$. Then both the numerator and denominator of $F$ are zero at $q=1$.
Nonetheless, we can apply L'H{\^o}pital's rule to calculate $\lim_{q\to 1}\frac{D(q,n)}{1-q}$ to deduce the expression in \eqref{eq:limatone}.

Now consider $\lim_{n \rightarrow \ell(q)} F(q,n)$. Suppose first $0 < q<1$.
Then
\[ D(q, \ell(q)) = 2(1-R)q(1-q) \left\{ [(1-R)q+R] + \frac{R^2(b_2-1)}{(1-R)q+R} \right\} \]
which is non-zero and has $\sgn(D(q, \ell(q))) = \sgn(1-R)$. It follows that for $q<1$, and $R<1$, $\lim_{n \uparrow \ell(q)} F(q,n) = - \infty$ and for $q<1$ and $R>1$, $\lim_{n \downarrow \ell(q)} F(q,n) = + \infty$.

Now suppose $q>1$, and if $R>1$ that $(1-R)q+R>0$. Then
\begin{align*}
D(q,\ell(q))   
&=2b_{1}[(1-R)q+R]\left(\frac{1-R}{b_{1}}q(1-q)+\frac{(b_{2}-1)R(1-R)}{b_{1}}\frac{q}{(1-R)q+R}\right) \\
&\qquad-2(1-R)q(1-q)[(1-R)q+R]-2R(1-R)(b_{2}-1)q\\
&=0.
\end{align*}
Then, in order to determine the value of $F(q, \ell(q))$ via L'H\^{o}pital's rule we need
\begin{equation}
\frac{\partial D}{\partial n}=2b_{1}[(1-R)q+R]-q\frac{\partial v}{\partial n}
=2b_{1}[(1-R)q+R]-b_{1}q\left(1-\frac{\sgn(1-R)\varphi}{\sqrt{\varphi^{2}+E^{2}}}\right).
\label{eq:dDdn}
\end{equation}
It follows that
\[
\left. \frac{\partial}{\partial n} D(q,n) \right|_{n = \ell} = 2 b_1 [(1-R)q + R] \left[ 1 - \frac{q[(1-R)q + R]}{[(1-R)q + R]^2 + R^2(b_2-1)} \right] \]
and hence we obtain \eqref{eq:limatn}.

(4) We prove the results for $R<1$. The results for $R>1$ can be obtained similarly, the only issue being that sometimes there is an extra case which arises when $(1-R)q + R$ changes sign. 

Note that for fixed $q$, the ordering of $m(q)$ and $\ell(q)$ is given by Part 1 of Lemma \ref{lemma:list}. The monotonicity of $F$ in $n$ for $q=1$ can be  obtained from \eqref{eq:limatone}.

If $0<q<1$, then since
\begin{align*}
2b_{1}[(1-R)q+R]-b_{1}q\left(1-\frac{\sgn(1-R)\varphi}{\sqrt{\varphi^{2}+E^{2}}}\right)>2b_{1}[(1-R)q+R]-2b_{1}q=2Rb_{1}(1-q)>0,
\end{align*}
we conclude from \eqref{eq:dDdn} that $D(q,n)$ is increasing in $n$. Since $D(q,m(q))=0$ it follows that $D(q,n)>0$ for $n>m(q)$ and $D(q,n)<0$ for $n<m(q)$. Hence, $F(q,n)=0$ if and only if $n=m(q)$, and we have
\begin{align*}
\sgn(F(q,n))&=-\sgn\left(\frac{D(q,n)}{(1-q)[(1-R)q+R][\ell(q)-n]}\right) \\
&=\sgn\left[(n-m(q))(n-\ell(q))\right].
\end{align*}
This gives the desired sign properties of $F(q,n)$ on the range $0<q<1$.

Now consider the case $q>1$. From Part 3 of this proof, we have $D(q,\ell(q))=0$. We can compute the second derivative of $D$ with respect to $n$ as
\begin{align*}
\frac{\partial^{2} D}{\partial n^{2}}&=\sgn(1-R)b_{1}^{2}q\frac{E^{2}}{(E^{2}+\varphi^{2})^{3/2}}
\end{align*}
so that (recall $R<1$) $D(q,n)$ is convex in $n$.
Since $D(q,m(q))=D(q,\ell(q))=0$, it follows that on the regime of $q>1$ we must have $D(q,n)<0$ when $n$ lies between $m(q)$ and $\ell(q)$ and $D(q,n)>0$ otherwise.
Thus $\sgn(D(q,n))=\sgn\left[(n-m(q))(n-\ell(q))\right]$. Then
\begin{align*}
\sgn(F(q,n))&=\sgn\left(\frac{D(q,n)}{\ell(q)-n}\right) =-\sgn(n-m(q)).
\end{align*}
Finally, note that $F(q,n)$ can be zero only if $n=m(q)$ or $n=\ell(q)$. But for $q>1$ the limiting expression at $n=\ell(q)$ is given by Part 3 of Lemma \ref{lemma:list}. Hence $F(q,n)=0$ if and only if $n=m(q)$. 

\end{proof}

The following lemma on further properties of $F$ is key in the proofs of the monotonicity property of $\Sigma$ and in results on comparative statics:
\begin{lemma} For $q \in (0,1]$ and $(1-R)m(q) < (1-R) n <(1-R) \ell(q)$, and for $q > 1$ and $(1-R)m(q) < (1-R) n$, we have
$\frac{\partial}{\partial n}F(q,n)\leqslant0$.
\label{lem:dFdn}
\end{lemma}

\begin{proof}

Direct computation gives
\[ 
(\ell(q)-n)^2 \frac{\partial }{\partial n}\left( \frac{D(q,n)}{\ell(q)-n} \right)
=(\ell(q)-n)\frac{\partial D}{\partial n}+D(q,n)
= A(q,n)
\]
where $A$ is defined in \eqref{eq:Adef}.
Differentiating $A$ we have
\begin{align*}
\frac{\partial }{\partial n}A(q,n)&= \sgn(1-R) \frac{b_{1}^{2}E(q)^{2}q(\ell(q)-n)}{(\varphi^{2}+E(q)^{2})^{3/2}}.
\end{align*}
Hence for $q>0$ and $R<1$, $A(q,n)$ is increasing in $n$ for $n<\ell(q)$ and decreasing in $n$ for $n>\ell(q)$. 
If $R>1$, then $A(q,n)$ is decreasing in $n$ for $n<\ell(q)$ and increasing in $n$ for $n>\ell(q)$.

Now we calculate the limiting value of $A(q,n)$ as $n\to\pm\infty$. Clearly $\varphi(q,n)\to\pm\infty$ as $n\to\pm\infty$. Then,
\begin{align*}
\lim_{(1-R)n\to+\infty}v(q,n)=\lim_{(1-R)\varphi\to+\infty}\varphi- \sgn(1-R)\sqrt{\varphi^{2}+E(q)^{2}}=0
\end{align*}
and
\begin{align*}
\lim_{(1-R)n\to+\infty}(\ell(q)-n)\left(1- \sgn(1-R)\frac{\varphi(n,q)}{\sqrt{\varphi(n,q)^{2}+E(q)^{2}}}\right)&=0.
\end{align*}
Observe that
\[
A(q,n) 
=2b_{1}[(1-R)q+R](\ell(q)-m(q))-b_{1}q(\ell(q)-n)\left(1- \sgn(1-R) \frac{\varphi}{\sqrt{\varphi^{2}+E^{2}}}\right) -qv(q,n)+qv(q,m(q))
\]
and thus
\[ 
\lim_{(1-R)n\to+\infty}A(q,n) = 2b_{1}[(1-R)q+R](\ell(q)-m(q))+qv(q,m(q))
=2(1-R)[(1-R)q+R]q(1-q).
\] 

Now we compute the limiting value of $A(q,n)$ as $\sgn(1-R) n\to-\infty$. In this case $v(q,n)$ is no longer converging. But consider
\begin{eqnarray*}
\lefteqn{ b_{1}q(\ell(q)-n)\left(1- \sgn(1-R)\frac{\varphi}{\sqrt{\varphi^{2}+E^{2}}}\right)+qv(q,n) } \\
&=& b_{1}q\ell(q)+q(\varphi-b_{1}n)- \sgn(1-R)\frac{q\varphi}{\sqrt{\varphi^{2}+E^{2}}}\left(b_{1}\ell(q)+(\varphi-b_{1}n)+\frac{E^{2}}{\varphi}\right).
\end{eqnarray*}
Using the fact that $\varphi-b_{1}n$ is independent of $n$, we can obtain
\begin{eqnarray*}
\lefteqn{ \lim_{(1-R)n\to-\infty}b_{1}q(\ell(q)-n)\left(1- \sgn(1-R)\frac{\varphi}{\sqrt{\varphi^{2}+E^{2}}}\right)+qv(q,n) } \\
&=&2b_{1}q\ell(q) - 2q\left[b_{1}-(1-R)(b_{3}-2R)q-(2-b_{2})R(1-R)\right]
\end{eqnarray*}
and thus
\begin{align*}
\lim_{(1-R)n\to-\infty}A(q,n)&=2b_{1}[(1-R)q+R](\ell(q)-m(q))+qv(q,m(q))-2b_{1}q\ell(q) \\
&\qquad+2q\left[b_{1}-(1-R)(b_{3}-2R)q-(2-b_{2})R(1-R)\right] \\
&=\frac{2R^{2}(1-R)(b_{2}-1)q(1-q)}{(1-R)q+R}
\end{align*}
after some algebra.

Suppose $R<1$.
For $0<q<1$ we have $A(q,n)$ increasing in $n$ for $n<l(q)$ and decreasing in $n$ for $n>\ell(q)$. Since on this range of $q$ $\lim_{n\to+\infty}A(q,n)=2(1-R)[(1-R)q+R]q(1-q)>0$ and $\lim_{n\to-\infty}A(q,n)=\frac{2R^{2}(1-R)(b_{2}-1)q(1-q)}{(1-R)q+R}>0$, we conclude $A(q,n)>0$ for all $n$.

If $q>1$ then $A(q,\ell(q))=D(q,\ell(q))=0$. But $A(q,n)$ attains its maximum at $n=\ell(q)$, hence we have $A(q,n)\leqslant 0$ for $q>1$. 
Putting the cases together, $(1-q) A(q,n) \geq 0$ and $\frac{\partial F}{\partial n} \leq 0$.

Now suppose $R>1$. Suppose $0<q<1$ or $q>R/(R-1)$. Then $A(q,n)$ decreasing in $n$ for $n<l(q)$ and increasing in $n$ for $n>\ell(q)$. Since $\lim_{n\to+\infty}A(q,n)<0$ and $\lim_{n\to-\infty}A(q,n)<0$, we conclude $A(q,n)<0$ for all $n$.
If $1<q<\frac{R}{R-1}$ then $A(q,n)$ attains its minimum of zero at $n=\ell(q)$. Hence $A(q,n)\geqslant 0$ for $1<q<\frac{R}{R-1}$.
Again we find
$(1-R)(1-q) A(q,n) \geq 0$ and $\frac{\partial F}{\partial n} \leq 0$.

It remains to check the result at $q=1$ and, if $R>1$, at $q = \frac{R}{R-1}$. At $q=1$ the result follows from \eqref{eq:limatone}. For $q=\frac{R}{R-1}$, using \eqref{eq:denO} we have
\[ F \left( \frac{R}{R-1}, n \right) = \frac{ (R-1)^2 D \left( \frac{R}{R-1}, n \right) }{2 R^3 (b_2-1)} \]
and the monotonicity of $F(\frac{R}{R-1},n)$ in $n$ follows from the monotonicity of $D( \frac{R}{R-1}, n)$ in $n$.

\end{proof}

\begin{proof}[Proof of Lemma \ref{lemma:onto}]
For any $u\in(0,q_{M})$, then since $(1-R)n_u(q)$ is decreasing in $q$ and $n'(\zeta(u))=0$, $n_{u}(q)$ can only cross $m(q)$ at some $q\geqslant q_{M}$.
Moreover, for $u \leq q \leq \zeta(u)$,  $(1-R)m(u) = (1-R)n_u(u) \geq (1-R)n_u(q) \geq (1-R)n_u(\zeta(u)) = (1-R) m(\zeta(u)) \geq (1-R)m_M$.

Since $n_{q_{M}}(q_{M})=m_{M}$, we have $\lim_{u\uparrow q_{M}}m(\zeta(u))=m_{M}$ and in turn $\lim_{u\uparrow q_{M}}\zeta(u)=q_{M}$. Then $\lim_{u\uparrow q_{M}}\Sigma(u)=0$.

Now consider $\Lambda(u):=\ln(1+\Sigma(u))=\int_{u}^{\zeta(u)}-\frac{R}{(1-R)q}\frac{O(q,n_{u}(q))}{n_{u}(q)}dq$. From the fact that $O(u, n_u(u))=O(u,m(u))= 0 = O(\zeta(u),m(\zeta(u))) =  O(\zeta(u), n_u(\zeta(u)))$ we have
\begin{align*}
\frac{d\Lambda}{du}=\int_{u}^{\zeta(u)}-\frac{R}{(1-R)q}\left(\frac{\partial}{\partial n}\frac{O(q,n_{u}(q))}{n_{u}(q)}\right)\frac{\partial n_{u}(q)}{\partial u}dq<0
\end{align*}
where we have used Lemma \ref{lem:dFdn} and the monotonicity of $n$ to make the conclusion about the sign.

We now show that $\lim_{u\downarrow 0}\Sigma(u)=+\infty$. We assume $R<1$; the proof for $R>1$ is similar. Consider a quadratic function $H(x)=(1-R)(m'(0)-x)-R(l'(0)-x)x$. Then trivially $H(m'(0))>0$. Choose a constant $k$ such that $m'(0)<k<\alpha<0$ where $\alpha$ is the negative root of $H(x)=0$. Then $H(k)>0$ and equivalently $k<\frac{(1-R)(m'(0)-k)}{R(l'(0)-k)}$. Now let $b(q)=1+kq$.
It is clear from the definition of $D$ that $D(0,1)=0$ and then
\[ \left. \frac{d}{dq} D(q, 1 + kq) \right|_{q=0}  = \left. \frac{\partial }{\partial q} D(q, n) \right|_{q=0,n=1} + k \left. \frac{\partial }{\partial n} D(q, n) \right|_{q=0,n=1}
 =  - 2 R b_1 m'(0) + 2 R b_1 k \]
and
\begin{align*}
\lim_{q\downarrow 0}O(q,b(q))= - \frac{(1-R)\frac{d}{dq} D(q, 1 + kq) |_{q=0} }{2 R^2 b_1 [\ell'(0) - k]}          =\frac{(1-R)(m'(0)-k)}{R(l'(0)-k)}.
\end{align*}
Then for all $\epsilon>0$, there exists $K_{\epsilon} \in (0,1)$ such that $O(q,b(q))>\frac{(1-R)(m'(0)-k)}{R(l'(0)-k)}-\epsilon$ for $q<K_{\epsilon}$. Choose $\epsilon$ such that $0<\epsilon<\frac{(1-R)(m'(0)-k)}{R(l'(0)-k)}-k$. Then we have $O(q,b(q))>k$ on $0<q<K_{\epsilon}$  and solutions to $n' = O(q,n)$ cross $b(q)$ from below. Let $\psi_{u}=\inf(q\geqslant u:n_{u}(q)>b(q))$. Then for $u<q<K_{\epsilon}\wedge \psi_{u}$, $n_{u}'(q)=O(q,n_{u}(q))>O(q,b(q))>k$. Moreover, there also exists $K_{m}$ such that $m'(q)<\frac{1}{2}(m'(0)+k)$ for $q<K_{m}$. Hence on $u<q<K_{\epsilon}\wedge \psi_{u}\wedge K_{m}$, $n'_{u}(q)-m'(q)>k-\frac{1}{2}(m'(0)+k)=\frac{1}{2}(k-m'(0))=:\widehat{k}>0$ and then $n_{u}(q)-m(q)>\widehat{k}(q-u)$. On the other hand, for $\psi_{u}<q<K_{\epsilon}\wedge K_{m}$, $m(q)<1+\frac{q}{2}(m'(0)+k)$ and hence $n_{u}(q)-m(q)>(1+kq)-(1+\frac{q}{2}(m'(0)+k))=\widehat{k}q>\widehat{k}(q-u)$. We conclude $n_{u}(q)-m(q)>\widehat{k}(q-u)$ for $u<q<Q:=K_{\epsilon}\wedge K_{m}$.

Hence, using \eqref{eq:formOalt}
and L'H{\^o}pital's rule,
\begin{align*}
\ln(1+\Sigma(u))&=\int_{u}^{\xi(u)}-\frac{R}{(1-R)q}\frac{O(q,n_{u}(q))}{n_{u}(q)}dq \\
&>\int_{u}^{Q}\frac{2b_{1}\left[(1-R)q+R\right](n_{u}(q)-m(q))-q\left[v(q,n_{u}(q))-v(q,m(q))\right]}{2q(1-q)[(1-R)q+R]b_{1}[l(q)-n_{u}(q)]}dq.
\end{align*}
For the denominator, and for $u<q<Q \leq 1$ we have
\begin{align*}
2q(1-q)[(1-R)q+R]b_{1}[l(q)-n_{u}(q)]&<2q(1-q)[(1-R)q+R]b_{1}[l(q)-m(q)] \\
&=2q^{2}(1-q)\{(1-R)(1-q)[(1-R)q+R]+(b_{2}-1)R(1-R)\} \\
&<2q^{2}\{M+(b_{2}-1)R(1-R)\}
\end{align*}
where $M:=\sup_{0<q<1}(1-R)(1-q)[(1-R)q+R]$.
For the numerator, note that for $q < \zeta(u)$
\begin{align*}
&v(q,n_{u}(q))-v(q,m(q)) \\
&=\varphi(q,n_{u}(q))-\varphi(q,m(q))-\{\sqrt{\varphi(q,n_{u}(q))^{2}+E(q)^{2}}-\sqrt{\varphi(q,m(q))^{2}+E(q)^{2}}\} \\
&<\varphi(q,n_{u}(q))-\varphi(q,m(q)) \\
&=b_{1}(n_{u}(q)-m(q)).
\end{align*}
Then,
\begin{align*}
&2b_{1}[(1-R)q+R](n_{u}(q)-m(q))-q\left[v(q,n_{u}(q))-v(q,m(q))\right] \\
&>\{2b_{1}[(1-R)q+R]-b_{1}q\}(n_{u}(q)-m(q)) \\
&=b_{1}L(q)(n_{u}(q)-m(q))
\end{align*}
where $L(q):=\{2[(1-R)q+R]-q\}$. Since $L$ is linear and $L(0)=2R>0$, we can choose to work on a small interval $(0,q_{L})$ such that $L(q)>\min(2R,L(q_{L}))>0$. For sufficiently small $u$ such that $u<q_{L}$, we have $b_{1}L(q)(n_{u}(q)-m(q))>b_{1}\min(2R,L(q_{L}))\widehat{k}(q-u)$ on $u<q<Q\wedge q_{L}$.

Putting everything together and setting $\widehat{Q}:=Q\wedge q_{L}\wedge 1$, for $u< \widehat{Q}$ we deduce that
\begin{align*}
\ln(1+\Sigma(u))&>\int_{u}^{\widehat{Q}}\frac{b_{1}\min(2R,L(q_{L}))\widehat{k}(q-u)}{2q^{2}[M+(b_{2}-1)R(1-R)]}dq \\
&=\frac{b_{1}\min(2R,L(q_{L}))\widehat{k}}{2[M+(b_{2}-1)R(1-R)]}\left(\ln\frac{\widehat{Q}}{u}+\frac{u}{\widehat{Q}}-1\right)
\end{align*}
Letting $u\downarrow 0$ and noting that $\widehat{Q}$ does not depend on $u$ we conclude that $\Sigma(u)\to\infty$.

\end{proof}

\section{Comparative Statics}
\label{app:compstat}
\begin{proof}[Proof of Proposition \ref{prop:compstat}]


(1) Set $\overline{m}(q) = b_1(m(q)-1)$ and similarly $\overline{n}(q) = b_1(n(q)-1)$ and $\overline{\ell}(q) = b_1(\ell(q)-1)$. The idea behind this transformation is that $\overline{m}$ is constructed such that it does not depend on $b_1$. $\overline{\ell}$ has a similar property. The free boundary value problem can be written as to find $(\overline{n},q_*,q^*)$ such that $\overline{n}'= \overline{O}(q,\overline{n})$ subject to $\overline{n}(q_*) = \overline{m}(q_*)$ and $\overline{n}(q^*) =\overline{m}(q^*)$.
Here $\overline{O}(q, \overline{n}) : = b_1 O(q, \frac{\overline{n}}{b_1}+1) = b_1 O(q,n)$.

Note that $\zeta(u)=\inf \{ q\geqslant u: (1-R)n_{u}(q)< (1-R)m(q) \}=\inf \{ q\geqslant u:(1-R)\overline{n}_{u}(q)<(1-R)\overline{m}(q) \}$.

Define $\overline{\varphi}(q,\overline{n}) = \varphi(q,n) = \varphi(q, \frac{\overline{n}}{b_1}+1)$, $\overline{v}(q,\overline{n}) = v(q,n) = v(q, \frac{\overline{n}}{b_1}+1)$ and $\overline{D}(q,\overline{n}) = D(q,n) = D(q, \frac{\overline{n}}{b_1}+1)$. Then, as functions of $q$ and $\overline{n}$, $\overline{\varphi}$, $\overline{v}$ and $\overline{D}$ are all independent of $b_1$.

We have
\[ \overline{O}(q,\overline{n})=-\frac{(1-R)(\overline{n}+b_{1})\overline{D}(q,\overline{n})}{2R(1-q)[(1-R)q+R][\overline{\ell}(q)-\overline{n}]} \]
By the above remarks the only dependence on $b_1$ is through the term $(\overline{n}+b_{1})$. Further
\[ \overline{n}' = (\overline{n}+b_{1}) \overline{F}(q,\overline{n}) \]
where $\overline{F}$ given by
\[ \overline{F}(q,\overline{n}) = F(q,n) = -\frac{(1-R)\overline{D}(q,\overline{n})}{2R(1-q)[(1-R)q+R][\overline{\ell}(q)-\overline{n}]} \]
does not depend on $b_1$. By Lemma~\ref{lem:dFdn}, $\overline{F}$ is decreasing in the second argument.

Let $\widehat{b}_{1}>\widetilde{b}_{1}$ be two positive values of $b_1$. Define $\widehat{n}_{u}$ and $\widetilde{n}_{u}$ the solutions to the initial value problem $\overline{n}'(q)=\overline{O}(q,\overline{n}(q))$ with $\overline{n}(u)=\overline{m}(u)$ under parameters $\widehat{b}_{1}$ and $\widetilde{b}_{1}$ respectively. We extend this notation to $O$, $\zeta$, $\Sigma$ and $(q_{*},q^{*})$ in a similar fashion.

If $\overline{n}_{u}$ is a solution to the initial value problem with $\overline{n}_u(u) = \overline{m}(u)$ we must have $(1-R)\overline{O}(q,\overline{n}_{u}(q))<0$ and hence $(1-R)\overline{O}$ is decreasing in $b_{1}$. Then $(1-R)\widehat{n}_{u}$ cannot upcross $(1-R)\widetilde{n}_{u}$ and since $(1-R)\widehat{n}_{u}'(u)=(1-R)\widehat{O}(u,\widehat{n}_{u}(u))<(1-R)\widetilde{O}(u,\widetilde{n}_{u}(u))=(1-R)\widetilde{n}_{u}'(u)$, we must have $(1-R)\widehat{n}_{u}(q)<(1-R)\widetilde{n}_{u}(q)$ at least up to $q=\widehat{\zeta}(u)\wedge\widetilde{\zeta}(u)$. From this we conclude $\widehat{\zeta}(u)<\widetilde{\zeta}(u)$. On the other hand,  $\overline{F}(q, \overline{n})$ depends on $b_{1}$ only through $\overline{n}$. It follows that
\[
- \ln(1+\overline{\Sigma}(u))
=\int_{u}^{\overline{\zeta}(u)} \frac{R}{q(1-R)}\frac{O(q, \frac{\overline{n}(q)}{b_1} + 1)}{\frac{\overline{n}(q)}{b_1} + 1} dq
=\int_{u}^{\overline{\zeta}(u)} \frac{R}{q(1-R)}\frac{\overline{O}(q,\overline{n}(q))}{\overline{n}(q)+b_{1}} dq
=\int_{u}^{\overline{\zeta}(u)} \frac{R}{q(1-R)} \overline{F}(q,\overline{n}(q)) dq.
\]
But, by the monotonity of $\overline{n}_u$ and $\overline{\zeta}$ in $b_1$
\[ \int_{u}^{\widehat{\zeta}(u)} \frac{R}{q(1-R)} \overline{F}(q,\widehat{n}(q))dq
> \int_{u}^{\widehat{\zeta}(u)} \frac{R}{q(1-R)} \overline{F}(q,\widetilde{n}(q)) dq
> \int_{u}^{\widetilde{\zeta}(u)} \frac{R}{q(1-R)} \overline{F}(q,\widetilde{n}(q))dq \]
where we use $(1-R) \overline{F}(q,n)<0$ and  and the fact that $\overline{F}$ is decreasing in $\overline{n}$ over the relevant range.
We conclude that
$\ln(1+\widehat{\Sigma}(u)) < \ln(1+\widetilde{\Sigma}(u))$ and hence
$\widehat{q}_{*}=\widehat{\Sigma}^{-1}(\xi)<\widetilde{\Sigma}^{-1}(\xi)=\widetilde{q}_{*}$.


To prove the monotonicity of the sale boundary $q^*$, one can parameterise the family of solutions via its right boundary point $(n_{v}(\cdot),\varsigma(v),v)$. See \cite{HobsonTseZhu:16} for the use of a similar idea.

(2) Now we consider the monotonicity of the limits of the no-transaction wedge in $b_3$. We use a different transformation and comparison result. Set $a(q)=n(q)-m(q)$. Then the original free boundary value problem becomes to solve $a'(q)=\underline{O}(q,a(q))$ subject to boundary conditions $a(q_{*})=a(q^{*})=0$ where
\begin{align*}
\underline{O}(q,a)=-\frac{(1-R)(a+m(q))D(q,a+m(q))}{2R(1-q)[(1-R)q+R]b_{1}[\ell(q)-m(q)-a]}-\frac{2R(1-R)}{b_{1}}q+\frac{b_{3}(1-R)}{b_{1}}.
\end{align*}
Observe that $b_{1}[\ell(q)-m(q)]=(1-R)q(1-q)+(b_{2}-1)R(1-R)\frac{q}{(1-R)q+R}$ does not depend on $b_{3}$.
Further, \begin{align*}
\varphi(q,a+m(q))&=b_{1}a+\varphi(q,m(q))=b_{1}a+R(1-R)\{(1-q)^{2}-(b_{2}-1)\}
\end{align*}
and $v(q,m(q))=-2R(1-R)(b_{2}-1)$ are both independent of $b_{3}$.
Hence $D(q,a+m(q))=2b_{1}[(1-R)q+R]a-q[v(q,a+m(q))-v(q,m(q))]$
and
\begin{align*}
\frac{O(q, a+m(q))}{a+m(q)} =  - \frac{(1-R)D(q,a+m(q))}{2R(1-q)[(1-R)q+R]b_{1}[\ell(q)-m(q)-a]}
\end{align*}
are independent of $b_{3}$. Recall we are assuming $R<1$. Then $O(q,n) \leq 0$ over the relevant range and
\begin{align*}
\frac{\partial \underline{O}}{\partial b_{3}}(q,a)&=-\frac{(1-R)D(q,a+m(q))}{2R(1-q)[(1-R)q+R]b_{1}[\ell(q)-m(q)-a]}\frac{\partial m}{\partial b_{3}}+\frac{1-R}{b_{1}} \\
&=-\frac{O(q,a+m(q))}{a+m(q)}\times\frac{1-R}{b_{1}}q+\frac{1-R}{b_{1}} > 0.
\end{align*}


Suppose $\widehat{b}_{3}>\widetilde{b}_{3}$. Using similar ideas in Part 1 of the proof we can deduce $\widehat{\zeta}(u)>\widetilde{\zeta}(u)$ and $\widehat{a}_{u}(q)>\widetilde{a}_{u}(q)$ for $q<\widetilde{\zeta}(u)$. 
Hence, using the fact that $\frac{O(q, a+m(q))}{a+m(q)}$ does not depend on $b_3$
\begin{align*}
\ln(1+\widehat{\Sigma}(u))&=\int_{u}^{\widehat{\zeta}(u)}\left(-\frac{R}{q(1-R)}\frac{O(q,\widehat{n}_{u}(q))}{\widehat{n}_{u}(q)}\right)dq \\
&=\int_{u}^{\widehat{\zeta}(u)}\left(-\frac{R}{q(1-R)}\frac{O(q,\widehat{a}_u(q)+m(q))}{\widehat{a}_{u}(q)+m(q)}\right)dq \\
&>\int_{u}^{\widetilde{\zeta}(u)}\left(-\frac{R}{q(1-R)}\frac{O(q,\widehat{a}_u(q)+m(q))}{\widehat{a}_u(q)+m(q)}\right)dq \\
&>\int_{u}^{\widetilde{\zeta}(u)}\left(-\frac{R}{q(1-R)}\frac{O(q,\widetilde{a}_u(q)+m(q))}{\widetilde{a}_u(q)+m(q)}\right)dq \\
&=\ln(1+\widetilde{\Sigma}(u))
\end{align*}
where we use the monotonicity of $\zeta(u)$ and  the property that $\frac{O(q,n)}{n}$ is decreasing in $n$ and hence $\frac{O(q,a+m(q))}{a+m(q)}$ is decreasing in $a$. Thus  $\widehat{q}_{*}=\widehat{\Sigma}^{-1}(\xi)>\widetilde{\Sigma}^{-1}(\xi)=\widetilde{q}_{*}$.
The monotonicity property of the sale boundary can be proved in a similar fashion by parameterising the family of solutions with their right boundary points.

\end{proof}

\begin{proof}[Proof of Theorem~\ref{thm:compstat2}]
(1) We write out the proof assuming $R<1$. The case $R>1$ follows similarly.

We use \eqref{eq:formO} to compute
\begin{align*}
\frac{\partial}{\partial b_{1}}O(q,n;b_{1})&=-\frac{2(1-R)^{2}qn/R}{\{2(1-R)(1-q)[(1-R)q+R]-\varphi(q,n)-\sqrt{\varphi(q,n)^{2}+E(q)^{2}}\}^{2}}\\
&\qquad \times\left(1+\frac{\varphi(q,n)}{\sqrt{\varphi(q,n)^{2}+E(q)^{2}}}\right)\frac{\partial \varphi}{\partial b_{1}}
\end{align*}
and hence, for $q>0$, $\sgn\left(\frac{\partial}{\partial b_{1}}O(q,n;b_{1})\right)=-\sgn\left(\frac{\partial \varphi}{\partial b_{1}}\right)=-\sgn(n-1)=+1$, since $n(\cdot)$ is bounded above by 1.

Further, $\overline{m}(q):=b_{1}(m(q;b_{1})-1)$ is independent of $b_{1}$ and from this we deduce $\frac{\partial}{\partial b_{1}}m(q;b_{1})=-\frac{m(q;b_{1})-1}{b_{1}}$ and hence over the continuation region $q\in[q_{*},q^{*}]$ we have $\sgn\left(\frac{\partial}{\partial b_{1}}m(q;b_{1})\right)=-\sgn(m(q;b_{1})-1)=+1$. Using the signs of $\frac{\partial}{\partial b_{1}}O(q,n;b_{1})|_{n=n(q)}$ and $\frac{\partial}{\partial b_{1}}m(q;b_{1})$ together with the fact that $q_{*}$ is decreasing in $b_{1}$, we conclude $n({}\cdot{};b_{1})$ is increasing in $b_{1}$. If we extend the domain of definition of $n$ to $[0,\infty)$ by setting $n(q)=n(q_{*})$ for $q<q_{*}$ and $n(q)=n(q^{*})$ for $q>q^{*}$ then we have $n({}\cdot{};b_{1})$ being increasing in $b_{1}$ on $[0,\infty)$.


Starting from the fact that $n(q;b_{1})$ is increasing in $b_{1}$, we can deduce that each of $-(1-q)N(q;b_{1})$, $hW(h;b_{1})$, $w(h;b_{1})$ and $(1-p)h'(p;b_{1})$ is increasing in $b_{1}$. Then for $\widehat{b}_{1}>\widetilde{b}_{1}$ (and using the overscripts to label the functions and parameters under the corresponding choice of $b_{1}$), we have
\begin{align}
\sgn(1-p)\widehat{h}'(p)>\sgn(1-p)\widetilde{h}'(p).
\label{eq:upcrossonly}
\end{align}

Recall that $G(p)=n(q_{*})^{-R}(1+\lambda p)^{1-R}$ and $G(p)=n(q^{*})^{-R}(1-\gamma p)^{1-R}$ on the purchase and sale region respectively. Using the monotonicity of $n$ in $b_{1}$ we conclude $\widehat{G}(p)<\widetilde{G}(p)$ over $p\in(0,\widehat{p}_{*})\cup(\widetilde{p}^{*},1/\gamma)$.

Suppose $G(p;b_{1})$ is not decreasing in $b_{1}$. Then since $G$ is continuous, $\widehat{G}(p)$ must cross $\widetilde{G}(p)$ at least twice, with the first cross being an upcross and the last cross being a downcross. Denote the $p$-coordinate of the first upcross and last downcross by $k_{u}$ and $k_{d}$ respectively.

Away from $p=1$, \eqref{eq:upcrossonly} implies that $\widehat{G}(p)$ cannot downcross $\widetilde{G}(p)$. Then the only possibility is that there are precisely two crossings with $0<k_{u}<k_{d}=1$. But if $k_{d}=1$ such that $K:=\widehat{G}(1)=\widetilde{G}(1)$, the relationship $\frac{1}{G(1)}\left(G(1)-\frac{G'(1)}{1-R}\right)^{1-1/R}=n(1)$ gives
\[ 
\widehat{G}'(1)=(1-R)\left(K-(K\widehat{n}(1))^{-R/(1-R)}\right)>(1-R)\left(K-(K\widetilde{n}(1))^{-R/(1-R)}\right)
=\widetilde{G}'(1)
\] 
contradicting the hypothesis that $k_{d}=1$ is a downcross.

(2) Now consider the monotonicity in $b_3$.
For $R<1$ a similar argument to the above can be applied if we can show that $n({}\cdot{};b_{3})$ is decreasing in $b_{3}$. But this follows immediately as $\sgn\left(\frac{\partial}{\partial b_{3}}O(q,n;b_{3})\right)=-\sgn\left(\frac{\partial \varphi}{\partial b_{3}}\right)=-1=\sgn\left(\frac{\partial}{\partial b_{3}}m(q;b_{3})\right)$ and $q_{*}$ is increasing in $b_{3}$.

For $R>1$ we cannot use this argument. However, the monotonicity of the value function in $b_3$, and hence the monotonicity of $\mathcal{C}$ can be proved by a comparison argument. The value function only depends on the parameters through $R$ and the auxiliary parameters, so when comparing two models which differ only through $b_3$ we may equivalently compare two models which differ only in $\alpha$.

Consider a pair of models, the only difference being that in the first model $Y$ has drift $\tilde{\alpha}$, whereas in the second model $Y$ has drift $\hat{\alpha}$ where $\hat{\alpha} > \tilde{\alpha}$. Write $\epsilon = \hat{\alpha} - \tilde{\alpha}>0$. Suppose that parameters are such that Standing Assumption~\ref{ass:sa}
holds in the first model; then necessarily Standing Assumption~\ref{ass:sa} holds in the second model. Let $(\tilde{Y},\hat{Y}) = (\tilde{Y}_t,\hat{Y}_t)_{t \geq 0}$
be given by
\[ (\tilde{Y}_t,\hat{Y}_t) = (ye^{\eta W_t + (\tilde{\alpha} - \frac{\eta^2}2) t}, ye^{\eta W_t + (\hat{\alpha} - \frac{\eta^2}2) t}) \]
so that $\hat{Y}_t = e^{\epsilon t} \tilde{Y}_t$. Let $(\tilde{C}, \tilde{\Pi}, \tilde{\Theta} = \theta + \tilde{\Phi} - \tilde{\Psi})$ be an admissible strategy for an agent in the first model. Suppose $\tilde{\Theta}$ is non-negative, and note that the optimal strategy has this property, even if the initial endowment in the illiquid asset is negative, since in that case there is an initial transaction into the no-transaction wedge which is contained in the half-plane $\theta \geq 0$. We may assume we start in the no-transaction region. Then $\tilde{X}_0 = x$ and $\tilde{X}= (\tilde{X}_t)_{t \geq 0}$ solves
\[
d\tilde{X}_{t} = r(\tilde{X}_t - \tilde{\Pi}_{t})dt +  \frac{\tilde{\Pi}_{t}}{S_{t}}dS_{t} -\tilde{C}_{t}dt-\tilde{Y}_{t}(1+\lambda)d\tilde{\Phi}_{t}+\tilde{Y}_{t}(1-\gamma)d\tilde{\Psi}_{t} .
\]
Define the absolutely continuous, increasing process $\kappa$ by $\kappa_t = \int_0^t \left\{ d \tilde{\Phi}_s \wedge ( d \tilde{\Psi}_s + \epsilon \tilde{\Theta}_s ds ) \right\}$ and set
\begin{eqnarray*}
\hat{\Pi}_t & = & \tilde{\Pi}_t \\
\hat{\Theta}_t & = & \tilde{\Theta}_t e^{-\epsilon t} \\
\hat{C}_t & = & \tilde{C}_t + (\lambda + \gamma) \tilde{Y}_t d\kappa_t + (1-\gamma) \epsilon \tilde{\Theta}_t \tilde{Y}_t dt \\
\hat{\Phi}_t & = & \int_0^t e^{-\epsilon s} \left( d\tilde{\Phi}_s - d \kappa_s \right) \\
\hat{\Psi}_t & = & \int_0^t e^{-\epsilon s} \left( d\tilde{\Psi}_s + \epsilon \tilde{\Theta}_s ds - d \kappa_s \right)
\end{eqnarray*}
Then $\hat{\Theta}_t \hat{Y}_t = \tilde{\Theta}_t \tilde{Y}_t$ and the corresponding wealth process solves
\begin{eqnarray*}
d\hat{X}_{t} & = & r(\hat{X}_t - \hat{\Pi}_{t})dt +  \frac{\hat{\Pi}_{t}}{S_{t}}dS_{t} -\hat{Y}_{t}(1+\lambda)d\hat{\Phi}_{t}+\hat{Y}_{t}(1-\gamma)d\hat{\Psi}_{t} -\hat{C}_{t}dt \\
& = & r(\hat{X}_t - \tilde{\Pi}_{t})dt +  \frac{\tilde{\Pi}_{t}}{S_{t}}dS_{t} -\hat{Y}_{t} e^{-\epsilon t} (1+\lambda)[d\tilde{\Phi}_{t} - d\kappa_t] +\hat{Y}_{t} e^{-\epsilon t} (1-\gamma)[d\tilde{\Psi}_{t} + \epsilon \tilde{\Theta}_t dt - d\kappa_t] \\
&& -\tilde{C}_{t}dt - (1-\gamma) \epsilon \tilde{\Theta}_t \tilde{Y}_t dt - (\lambda+ \gamma) \tilde{Y}_t d \kappa_t \\
&=& r(\hat{X}_t - \tilde{\Pi}_{t})dt +  \frac{\tilde{\Pi}_{t}}{S_{t}}dS_{t} -\tilde{C}_{t}dt-\tilde{Y}_{t}(1+\lambda)d\tilde{\Phi}_{t}+\tilde{Y}_{t}(1-\gamma)d\tilde{\Psi}_{t} .
\end{eqnarray*}
If $\hat{X}_0 = x = \tilde{X}_0$ then $\hat{X}$ solves the same equation as $\tilde{X}$ and $\hat{X}_t = \tilde{X}_t \geq 0$. Then, for any admissible strategy in the first model for which $(\Theta_t)_{t \geq 0}$ is positive, including the optimal strategy in this model, there is a corresponding admissible strategy in the second model with strictly larger consumption at all future times. Hence the value function is strictly greater in the second model.
\end{proof}

\section{The consistency condition on transaction costs}
\label{app:tcc}
Fix positive constant $\epsilon>0$ and define $\delta_{1}(\epsilon):=1-W(h(1-\epsilon))>0$ and $\delta_{2}(\epsilon):=W(h(1+\epsilon))-1>0$. Then for
\begin{align*}
o(\epsilon, \delta_{1},\delta_{2}):=\ln\left(\frac{1-\epsilon}{1+\epsilon}\right)-\ln\left(\frac{\delta_{2}(1-\delta_{1})}{\delta_{1}(1+\delta_{2})}\right),
\end{align*}
we have for $p_{*}<1<p^{*}$
\begin{align*}
\ln(1+\xi)+o(\epsilon, \delta_{1}(\epsilon),\delta_{2}(\epsilon))&=\left[\int_{p_{*}}^{1-\epsilon}\frac{dp}{p(1-p)}+\int_{1+\epsilon}^{p^{*}}\frac{dp}{p(1-p)}\right]-\left[\int_{q_{*}}^{1-\delta_{1}(\epsilon)}\frac{dq}{q(1-q)}+\int_{1+\delta_{2}(\epsilon)}^{q^{*}}\frac{dq}{q(1-q)}\right] \\
&=\left[\int_{h_{*}}^{h(1-\epsilon)}\frac{dh}{w(h)}+\int_{h(1+\epsilon)}^{h^{*}}\frac{dh}{w(h)}\right]-\left[\int_{q_{*}}^{1-\delta_{1}(\epsilon)}\frac{dq}{q(1-q)}-\int_{1+\delta_{2}(\epsilon)}^{q^{*}}\frac{dq}{q(1-q)}\right] \\
&=\left[\int_{q_{*}}^{W(h(1-\epsilon))}\frac{N'(q)dq}{(1-R)qN(q)}+\int_{W(h(1+\epsilon))}^{q^{*}}\frac{N'(q)dq}{(1-R)qN(q)}\right]\\
&\qquad-\left[\int_{q_{*}}^{1-\delta_{1}(\epsilon)}\frac{dq}{q(1-q)}-\int_{1+\delta_{2}(\epsilon)}^{q^{*}}\frac{dq}{q(1-q)}\right] \\
&=\int_{q_{*}}^{1-\delta_{1}(\epsilon)}\left(-\frac{R}{q(1-R)}\frac{O(q,n(q))}{n(q)}\right)dq+\int_{1+\delta_{2}(\epsilon)}^{q^{*}}\left(-\frac{R}{q(1-R)}\frac{O(q,n(q))}{n(q)}\right)dq.
\end{align*}
On sending $\epsilon\downarrow 0$, we have $\delta_{1}(\epsilon)\downarrow 0$ and $\delta_{2}(\epsilon)\downarrow 0$ and thus
\begin{align*}
\int_{q_{*}}^{q^{*}}\left(-\frac{R}{q(1-R)}\frac{O(q,n(q))}{n(q)}\right)dq=\ln(1+\xi)+\lim_{\epsilon\downarrow 0}o(\epsilon, \delta_{1}(\epsilon),\delta_{2}(\epsilon)).
\end{align*}
Now,
\begin{align*}
\frac{\delta_{2}(\epsilon)}{\delta_{1}(\epsilon)}=\frac{W(h(1+\epsilon))-1}{1-W(h(1-\epsilon))}=\frac{W(h(1+\epsilon))-1}{\epsilon}\frac{\epsilon}{1-W(h(1-\epsilon))}.
\end{align*}
But
\begin{align*}
\frac{1-W(h(1-\epsilon))}{\epsilon}&=\frac{(1-R)h(1-\epsilon)-\epsilon(1-\epsilon)h'(1-\epsilon)}{(1-R)\epsilon h(1-\epsilon)}=1-\frac{(1-\epsilon)G'(1-\epsilon)}{(1-R)G(1-\epsilon)}
\end{align*}
and thus $\lim_{\epsilon\downarrow 0}\frac{1-W(h(1-\epsilon))}{\epsilon}=1-\frac{G'(1)}{(1-R)G(1)}$. 
Similarly, we have $\lim_{\epsilon\downarrow 0}\frac{W(h(1+\epsilon))-1}{\epsilon}=1-\frac{G'(1)}{(1-R)G(1)}$. Hence $\lim_{\epsilon\downarrow 0}o(\epsilon, \delta_{1}(\epsilon),\delta_{2}(\epsilon))=0$ and \eqref{eq:trancost} holds. In case either $p_{*}=1$ or $p^{*}=1$, a similar argument can be used to show that \eqref{eq:trancost} is still valid.

\begin{thebibliography}{references}

\bibitem{HerczeghProkaj:15}
Herczegh A. and Prokaj V.
\newblock Shadow price in the power utility case.
\newblock {\em Annals of Applied Probability}, 25(5):2671--2707, 2015.

\bibitem{AkianMenaldiSulem:95}
M~Akian, Menaldi J.L., and Sulem A.
\newblock Multi-asset portfolio selection problem with transaction costs.
\newblock {\em Mathematics and Computers in Simulation}, 38(1-3):163--172,
  1995.

\bibitem{BichuchGuasoni:16}
M~Bichuch and P.~Guasoni.
\newblock Investing with liquid and illiquid assets.
\newblock SSRN: ssrn.com/abstract=2523538, 2016.

\bibitem{BichuchShreve:13}
M~Bichuch and S.E. Shreve.
\newblock Utility maximisation trading two futures with transactio costs.
\newblock {\em SIAM Journal of Mathamtical Finance}, 4(1):26--85, 2013.

\bibitem{Cadenillas:00}
A.~Cadenillas.
\newblock Consumption-investment problems with transaction costs: survey and
  open problems.
\newblock {\em Mathematical Methods of Operations Research}, 51:43--68, 2000.

\bibitem{Choi:16}
Jin~Hyuk Choi.
\newblock Optimal investment and conumption with liquid and illiquid assets.
\newblock ArXiV: arXiv preprint 1602.06998, 2016.

\bibitem{ChoiSirbuZitkovic:13}
Jin~Hyuk Choi, Mihai Sirbu, and Gordan Zitkovic.
\newblock Shadow prices and well-posedness in the problem of optimal investment
  and consumption with transaction costs.
\newblock {\em SIAM Journal on Control and Optimization}, 51(6):4414--4449,
  2013.

\bibitem{CollingsHausmann:99}
P.~Collings and Hausmann U.G.
\newblock Optimal portfolio selection with transaction costs.
\newblock {\em In: Proceedings of the Conference on Control of Distributed and
  Stochastic Systems, Hangzhou, China}, pages 189--197, 1999.

\bibitem{ConstantinidesMagill:76}
G.M. Constantinides and M.J.P. Magill.
\newblock Portfolio selection with transaction costs.
\newblock {\em Journal of Economic Theory}, 13:264--271, 1976.

\bibitem{DaiZhong:10}
M~Dai and Y.~Zhong.
\newblock Penalty methods for continuous-time portfolio selection with
  proportional transaction costs.
\newblock {\em Journal of Computational Finance}, 13(3):1--31, 2010.

\bibitem{DavisNorman:90}
Mark~HA Davis and Andrew~R Norman.
\newblock Portfolio selection with transaction costs.
\newblock {\em Mathematics of Operations Research}, 15(4):676--713, 1990.

\bibitem{DumasLuciano:91}
B.~Dumas and E.~Luciano.
\newblock An exact solution to a dynamic portfolio choice problem under
  transaction costs.
\newblock {\em Journal of Finance}, 46:577--595, 1991.

\bibitem{EvansHendersonHobson:08}
J.D. Evans, Henderson V., and D.~Hobson.
\newblock Optimal timing for an indivisible asset sale.
\newblock {\em Mathematical Finance}, 18(4):545--567, 2008.

\bibitem{GuasoniMuhleKarbe:12}
P.~Guasoni and Muhle-Karbe J.
\newblock Portfolio choice with transaction costs: a user's guide.
\newblock {\em Available at SSRN 2120574}, 2012.

\bibitem{HobsonTseZhu:16}
David Hobson, Alex Sing~Lam Tse, and Yeqi Zhu.
\newblock Optimal consumption and investment under transaction costs.
\newblock ArXiV preprint arXiv:1612.00720, 2016.

\bibitem{HobsonZhu:14}
David Hobson and Yeqi Zhu.
\newblock Multi-asset consumption-investment problems with infinite transaction
  costs.
\newblock ArXiV preprint arXiv:1409.8307, 2014.

\bibitem{JanecekShreve:04}
K.~Janacek and Shreve S.E..
\newblock Asymptotic analysis for optimal investment and consumption with
  transaction costs.
\newblock {\em Finance and Stochastics}, 18(2):181--206, 2004.

\bibitem{KallsenMuhleKarbe:10}
J.~Kallsen and Muhle-Karbe J.
\newblock On using shadow prices in portflio optimization with transaction
  costs.
\newblock {\em Annals of Applied Probability}, 20(4):1341--1358, 2010.

\bibitem{Liu:04}
H.~Liu.
\newblock Optimal consumption and investment with transaction costs and
  multiple risky assets.
\newblock {\em The Journal of Finance}, 59(1):289--338, 2004.

\bibitem{Merton:69}
R.C. Merton.
\newblock Life portfolio selection under uncertainty: the continuous-time case.
\newblock {\em The Review of Economics and Statistics}, 51:247--257, 1969.

\bibitem{MuthurmanKumar:06}
K.~Muthuraman and S.~Kumar.
\newblock Multi-dimensional portfolio optimisation with proportional
  transaction costs.
\newblock {\em Mathematical Finance}, 16(2):301--335, 2006.

\bibitem{PossamaiSonerTouzi:15}
D.~Possama\"{i}, H.M. Soner, and N.~Touzi.
\newblock Homogenization and asymptotics for small transaction costs: the
  multidimensional case.
\newblock {\em Communications in Partial Differential Equations}, 40:609--692,
  2015.

\bibitem{ShreveSoner:94}
S.E. Shreve and Soner H.M.
\newblock Optimal investment and consumption with transaction costs.
\newblock {\em Annals of Applied Probability}, 4:609--692, 1994.

\bibitem{SonerTouzi:13}
H.M. Soner and N.~Touzi.
\newblock Homogenization and asymptotics for small transaction costs.
\newblock {\em SIAM Journal of Control and Optimization}, 51(4):2893--2921,
  2013.

\bibitem{tse:16}
Alex Sing~Lam Tse.
\newblock {\em Dynamic economic decision problems under behavioural preferences
  and market imperfections}.
\newblock PhD thesis, University of Warwick, 2016.

\bibitem{WhalleyWilmott:97}
A.E. Whalley and Wilmott P.
\newblock An asymptotic analysis of an optimal hedging model for option pricing
  with transaction costs.
\newblock {\em Mathematical Finance}, 7(3):307--324, 1997.

\bibitem{ChenDai:13}
Chen X.F. and Dai M.
\newblock Characterisation of optimal strategy for multi-asset investment and
  consumption with transaction costs.
\newblock {\em SIAM Journal on Financial Mathematics}, 4(1):857--883, 2013.

\end{thebibliography}
\end{document}